\documentclass{article}
\pdfoutput=1
\usepackage{a4wide}
\usepackage[utf8x]{inputenc}
\usepackage{amsmath,amsfonts,amssymb,mathtools,amsthm,mathrsfs,bbm}
\usepackage{dsfont}
\usepackage{graphicx}
\usepackage{xcolor}
\usepackage{natbib}
\usepackage{bbding}
\usepackage{csquotes}
\usepackage{tikz}
\usetikzlibrary{decorations.pathreplacing}
\definecolor{darkgrey}{rgb}{0.6,0.6,0.6}
\usepackage{overpic}

\newtheorem{theorem}{Theorem}
\newtheorem{proposition}{Proposition}[section]

\newtheorem{corollary}[proposition]{Corollary}
\newtheorem{definition}[proposition]{Definition}
\theoremstyle{definition}
\newtheorem{remark}[proposition]{Remark}

\usepackage{microtype}

\makeatletter
\renewcommand{\@makecaption}[2]{%
\begin{quote}
\small {\bf #1}~#2
\end{quote}}
\makeatother

\usepackage[colorlinks,linkcolor=blue,citecolor=blue]{hyperref}
\hypersetup{pdfauthor={Franz Baumdicker},pdftitle={The site frequency spectrum of dispensable genes}}

\usepackage{etoolbox}
\makeatletter

\patchcmd{\NAT@citex}
  {\@citea\NAT@hyper@{%
     \NAT@nmfmt{\NAT@nm}%
     \hyper@natlinkbreak{\NAT@aysep\NAT@spacechar}{\@citeb\@extra@b@citeb}%
     \NAT@date}}
  {\@citea\NAT@nmfmt{\NAT@nm}%
   \NAT@aysep\NAT@spacechar\NAT@hyper@{\NAT@date}}{}{}

\patchcmd{\NAT@citex}
  {\@citea\NAT@hyper@{%
     \NAT@nmfmt{\NAT@nm}%
     \hyper@natlinkbreak{\NAT@spacechar\NAT@@open\if*#1*\else#1\NAT@spacechar\fi}%
       {\@citeb\@extra@b@citeb}%
     \NAT@date}}
  {\@citea\NAT@nmfmt{\NAT@nm}%
   \NAT@spacechar\NAT@@open\if*#1*\else#1\NAT@spacechar\fi\NAT@hyper@{\NAT@date}}
  {}{}

\makeatother

\begin{document}
\title{\Large The site frequency spectrum of dispensable genes}
% \author[ufb]{Franz Baumdicker}
% \ead{baumdicker@stochastik.uni-freiburg.de}
% \address[ufb]{Albert-Ludwigs University of Freiburg,
%      Eckerstr.~1, D--79104 Freiburg, Germany}
\author{Franz Baumdicker\thanks{Abteilung f\"ur
    Mathematische Stochastik, Albert-Ludwigs University of Freiburg,
    Eckerstr.~1, D--79104 Freiburg, Germany, e-mail: baumdicker@stochastik.uni-freiburg.de}
%\; and Peter  Pfaffelhuber$\mbox{}^{1,}$\thanks{e-mail:  p.p@stochastik.uni-freiburg.de}
}

\thispagestyle{empty}

\date{\today}

\maketitle

\begin{abstract}
The differences between DNA-sequences within a population are the basis to infer the ancestral relationship of the individuals.
Within the classical infinitely many sites model, it is possible to estimate the mutation rate based on the site frequency spectrum, 
which is comprised by the numbers $C_1,\dots,C_{n-1}$ where $n$ is the sample size and $C_s$ 
is the number of site mutations (Single Nucleotide Polymorphisms, SNPs) which are seen in $s$ genomes. 
Classical results can be used to compare the observed site frequency spectrum with its neutral expectation, $E[C_s]=\theta_2/s$, where $\theta_2$ is the scaled site mutation rate.
In this paper, we will relax the assumption of the infinitely many sites model that all individuals only carry homologous genetic material.
Especially, it is today well-known that bacterial genomes have the ability to gain and lose genes, such that every single genome is a mosaic of genes, and genes are present and absent
in a random fashion, giving rise to the dispensable genome. While this presence and absence has been modeled under neutral evolution within 
the infinitely many genes model in \cite{BaumdickerHessPfaffelhuber2010}, we link presence and absence of genes with the numbers of site mutations seen within each gene.
In this work we derive a formula for the expectation of the joint gene and site frequency spectrum, 
denotes $G_{k,s}$, the number of mutated sites occurring in exactly $s$ gene sequences, while the corresponding gene is present in exactly $k$ individuals. 
We show that standard estimators of $\theta_2$ for dispensable genes are biased and that the site frequency spectrum for dispensable genes differs from the classical result.

%% ERSTE VERSION
% The differences between DNA-sequences coding for a gene within a population hold a lot of information about
%  the ancestral relationship of the individuals.
% The well studied infinitely many sites model is one of the standard models to analyze this genetic variation.
% Within this model it is possible to estimate the mutation rate based on the site frequency spectrum, i.e.\ the number of mutated \emph{sites} 
% carried by a certain proportion of the population.
% So far the expected site frequency spectrum is only known for genes which can not be lost
% and thus are present in all individuals, which is a reasonable assumption for eukaryotic genes.
% For many prokaryotic genes this assumption is surely not true, 
% as bacteria have a more flexible genome, such that the majority of genes is only present in a subset of the population.
% In this work we introduce a model to analyze the site frequency spectrum for these dispensable prokaryotic genes, which in principle can get lost.
% We derive a formula for the expectation of the joint gene and site frequency spectrum $|G_{k,s}|$, the number of genes with $s$ mutated sites present in exactly $k$ individuals.
% We show that the site frequency spectrum for dispensable genes differs from the classical result.
\end{abstract}

% \begin{keyword}
% population genetics \sep site frequency spectrum \sep dispensable gene \sep pangenome \sep Tajima's~D 
% \end{keyword}



\section{Introduction}

% from single gene sequences to genomes

\subsubsection*{dispensable genes \dots}
%introduce dispensable genes
% Sequencing projects are no longer limited to a single gene sequence, as genome sequencing is meanwhile affordable for huge datasets, 
In the last decade more and more complete genomes were sequenced \citep{Binnewies2006}. %\cite{Liolios2010}. 
Comparison of whole genomes, especially in bacterial populations, revealed that not every
individual of a population carries the same set of genes.
There is a nonempty class of dispensable genes which are only present in a subset of the population.
The set of all genes present at least somewhere in the population, is thus larger than 
any single genome. In \cite{Medini2005} the term \emph{pangenome} was introduced to describe the set of genes distributed among a population.
In fact, it has been shown that dispensable genes constitute a large part of many bacterial pangenomes. \citep{Binnewies2006,Lapierre2009,Bentley2009,BaumdickerHessPfaffelhuber2011,Haegeman2012,Lobkovsky2013}.

The pangenome gives rise to the gene frequency spectrum, $G_1,\dots,G_n$.
If the genes of $n$ sampled individuals are known, $G_i$ is
the number of genes present in exactly $i$ of these $n$ individuals.
In \cite{Koonin2008} the pangenome is divided into three parts, based on the gene frequencies.
The genes present in all individuals constitute the core of the pangenome.
Genes in intermediate frequency are part of the shell, while the genes present in very few individuals belong to the cloud. 
It is reasonable to assume that most of the core genes have 
a function essential to survive and are under selection pressure \citep{Fang2005}. 
In contrast, the function of many dispensable genes is still unknown.

% The joint site frequency spectrum $\mathbb E[G_{k,s}]$

% introduce site frequencys and infinitely many sites model
% The inference of ancestral relationships is most often based on the observed differences between present individuals.
\subsubsection*{ \dots and the site frequency spectrum}
The function of a gene is encoded by its gene sequence, i.e. a string of DNA which is
ususally recorded as a string of the letters A, G, T and C.
Occasionally, during reproduction, a mutation changes a letter of the gene sequence.
Hence, the same gene can differ between individuals at any position (site) of the DNA strand.
The observed mutated sites (Single Nucleotide Polymorphisms, SNPs) of the gene sequences 
%common to all individuals 
hold a lot of information.
Statistics based on SNP pattern can be used to estimate population parameters and infer the history of the population.
In particular, the site frequency spectrum, $C_1,\dots,C_n$, where $C_s$ is the number of SNPs observed in exactly $s$ individuals, is frequently examined.
To investigate the differences between gene sequences caused by site mutations.
the infinitely many sites model \citep{Kimura1969} is one of the standard models.
Within the infinitely many sites model the expected site frequency spectrum $\mathbb E[C_s] = \frac{\theta_2}{s}$ , 
i.e. the mean number of SNPs that occur in exactly $s$ individuals, has originally been computed by \cite{Watterson1975},
but nice derivations are also shown in \citep{Fu1995} and \citep[chap.~9.4]{Ewens2004}.
A bunch of estimators for $\theta_2$, the rate at which site mutations occur, is based on the site frequency spectrum and enables us to test for neutral evolution \citep{Watterson1975,Tajima1983,Tajima1989a,Achaz2009}.
The infinitely many sites model assumes that a gene is present in all individuals at any time to model site mutations in the DNA sequences of this gene.
In this paper we ask:
\begin{quote}
How does the site frequency spectrum change if this assumption is violated and a dispensable gene sequence, which is only present in a subset of the population, is analyzed? 
\end{quote}

We consider the frequencies of site mutations within dispensable gene sequences.
Each mutated site in a dispensable gene has two assigned frequencies.
\begin{itemize}
 \item[(i)] The frequency $s$ of the SNP itself, which is the number of individuals with the corresponding mutation.
 \item[(ii)] The frequency $k$ of the corresponding gene sequence, 
which is the number of individuals which possess the gene.
\end{itemize}
\begin{figure}
 \centering
 \begin{tabular}{ccccccc}
	& gene 1 		& gene 2 	& gene 3 		& gene 4 		& \dots 	& gene $m$ \\
refseq	& \texttt{ATGTCT}    & \texttt{GCTATG} 	&\texttt{CCGTTGGAG}	&\texttt{TCGGAGCAG}	&		& \texttt{TGA} \\
ind. 1	& \texttt{--A-T-}    & $\varnothing$ 	&\texttt{-------T-}	&$\varnothing$		&		& $\varnothing$ \\
ind. 2	& \texttt{--A---}    & \texttt{--CT---} &\texttt{T--A-----}	& $\varnothing$		&		& \texttt{--T} \\
ind. 3	& $\varnothing$	     & $\varnothing$ 	&\texttt{T-----C--}	&$\varnothing$		&		& \texttt{C--} \\
ind. 4	& $\varnothing$      & $\varnothing$ 	&\texttt{T--A-----}	& \texttt{--A---A--}	&		& \texttt{C--} \\
\vdots	& \vdots		& \vdots 	&\vdots			&\vdots			&		&\vdots \\
ind. $n$& \texttt{-----C}    & $\varnothing$ 	&\texttt{T--A-----}	& $\varnothing$		&		& $\varnothing$ \\
 \end{tabular}
\caption{Data structure for joint distribution of segregating sites and dispensable genes. The aligned sequences of several dispensable genes are shown.
 The first row shows the reference sequence of the gene, i.e. the ancestral gene sequence at the time the gene first appeared in the population.
 The remaining rows show the site mutations of the current gene sequences in $n$ sampled individuals.
 If the whole gene is absent in an individual the symbol $\varnothing$ is used. \label{datatable}}
\end{figure}
Figure \ref{datatable} shows a typical data structure for site mutations within dispensable genes.

% outline
In Section \ref{modelgsf} we will introduce a population genetics framework for the joint evolution of the frequency of dispensable genes, 
driven by events of gene gain and gene loss, and the mutations within the corresponding sequences.
In Section \ref{gsfresults} the main goal is to calculate the joint frequency spectrum within this framework. 
We denote the joint site and gene frequency spectrum, i.e. the number of mutated sites occurring in 
exactly $s$ individuals, where the corresponding gene is present in exactly $k$ individuals,
% by $(G_{k,s})_{{1 \leq k \leq n\phantom{-1}} \atop {1 \leq s \leq k- 1}}$.
by $(G_{k,s})_{\genfrac{}{}{0pt}{}{1 \leq k \leq n\phantom{-1}} {1 \leq s \leq k- 1}}.$
The first moment of $G_{k,s}$ is given in Theorem \ref{jointfs}.
In particular, we are interested in the marginal mean site frequency spectrum for genes present in exactly $k$ individuals.
We will show in Corollary \ref{condsitefreqspec} that, the site frequency spectrum for sequences present in $k$ out of $n$ individuals
 differs from the site frequency spectrum based on a sample of $k$ individuals, where each individual possesses the gene sequence.

In Section \ref{effectonestimates} we further investigate the difference between site frequencies of dispensable and non-dispensable genes.
Using the classical infinitely many sites model to estimate site mutation rates, although at least some of the gene sequence considered are 
only present in a subset of the population, will lead to an underestimation.
To be more precise let us consider a single gene sequence found in $k$ out of $n$ individuals, where we are interested in the site mutation rate $\theta_2$.
We will show in Theorem \ref{dispestimates} that two frequently used estimators for the site mutation rate $\theta_2$, 
namely Watterson's Estimator $\widehat \theta_W$ \citep{Watterson1975} as well as Tajimas estimator $\widehat \pi$ \citep{Tajima1983}, will both 
have a negative bias which is at least as high as  $\frac{n-k}{n}\theta_2$,
such that $\mathbb E[\widehat \theta_W] \leq \frac{k}{n} \theta_2$ and $\mathbb E[\widehat \pi] \leq \frac{k}{n} \theta_2$.
We shortly discuss the impact of the shown results in Section~\ref{discussion} and illustrate the effect on estimates of Tajima's D \citep{Tajima1989a}.
Finally most of the proofs are given in Section~\ref{proofs}.

% % Introductionary text about dN/dS ratios
% A frequently used statistic to detect genes under selective pressure is the $\frac{d_N}{d_S}\text{-ratio}$ for a gene sequence.
% Due to the redundancy of the genetical code some of the DNA-substitutions within a sequence do not lead to a change in the corresponding protein sequence.
% Such substitutions are called synonymous.
% In contrast the non-synonymous substitution will change the protein sequence,
%  such that non-synonymous substitutions should have a much larger selective effect than synonymous substitutions.
% This is the basic idea behind the $\frac{dN}{dS}$ ratio teststatistic.
% A $\frac{dN}{dS}$ ratio of one should be obtained if the sequence is neutral, i.e there is no selection. While ratios greater(less) than one indicate positive (purifying) selection.
% \color{black}

% 
%  {\color{red}}
% The limit of the site frequency spectrum (gene frequency spectrum) for large sample sizes gives rise to the green function, which describes the number of sites (genes) at a given frequency.
% Within our model we will show an analogous result for the green function which describes the number of sites in genes in frequency $x$. % rauslassen
% 

\section{A model for gene and site frequency evolution}
\label{modelgsf}

\subsection{Wright--Fisher model}

Consider a population with a constant population of $N$ individuals evolving in a Wright--Fisher model.
That is, each of the $N$ individuals of the current generation chooses his parent at random 
among the individuals of the previous generation.

We model a pool of potential genes, which can be gained and lost again, by the set $I = [0,1]$, where each point corresponds to a potential gene.
In addition for each gene sequence $u \in I$ its DNA-sequence is modeled by the set $J = (0,1]$.
Each point $v \in J$ corresponds to a site, which can be hit by a mutation.
Note that we have excluded the point $0$ from $J$ to use this point as a marker for the presence and absence of genes.

% As mutations are assigned to $J$ according to the Lebesgue measure, we might as well use $J = [0,1]$ but for simplicity we will exclude $0$ from $J$ here.
Each individual carries a set of genes, which constitute the genome of this individual.
In addition each of the genes carries its own set of site mutations.
We will describe the genome and the gene sequences of an individual by a simple finite counting measure $m$ on $[0,1]^2$, such that $m(u,0)=1$ if gene $u$ is present in this individual and
$m(u,I)=0$ otherwise. In the same manner the gene $u$ in this individual carries a mutation at site $v$  iff $m(u,v) = m(u,0) = 1$.

The combined mutational mechanisms for genes and the sites within these genes now work as follows:\\
After an individual chooses its parent, but before it inherits the state of the parent the following 
mutational mechanisms change the genomic and the genetic state $m$ of the parent to $m'$ of the child.
\begin{description}
\item[genomic]
\item[gene gain] before reproduction of an individual with probability $\mu_1$ a new gene $u \in I$ is added to the genome $m$.
      The gene is assumed to be completely new to the population, i.e.\ $u$ is chosen uniform from $I$.
      If gene $u$ is gained, the state changes to $m' = m + \delta_{(u,0)}$.
\item[gene loss] each gene present in the parent is lost independently with probability $\nu$ during reproduction and no longer present in the child,
      i.e.\ with probability $1-\nu$ a gene present in the parent will still be present in the child.
      If gene $u$ is lost, the state changes to $m' = m - m|_{\{u\}\times I}$.
\item[genetic]
\item[site mutations] each gene $u$ present in the child, i.e. $m'(u,0) = 1$, suffers an additional site mutation with probability $\mu_2$.
      The mutation hits a uniform site $v \in J$, which has never been hit before, and changes the genetic state to $m' = m + \delta_{(u,v)}$.
\end{description}

\begin{remark}
 It is well known that the mutation rate of gene sequences may vary between different regions of the genome .
 While a straightforward adaptation for gene specific mutations rates should be possible,
 we will keep the model as simple as possible and assume that for any gene sequence the site mutation rate $\mu_2$ equals.
\end{remark}

\begin{remark}
 In this model we combined the mutational mechanisms of previous publications to model genetic and genomic variation at once.
 New genes may be gained and each present gene can get lost, just as in the infinitely many \emph{genes} model \citep{BaumdickerHessPfaffelhuber2010}.
And the sites within each of the genes can get hit by mutations, just as in the infinitely many \emph{sites} model \citep{Kimura1969}.
\end{remark}

\subsection{Kingman's coalescent}

Kingman's coalescent, given in Definition \ref{Kingmancoal}, was introduced by \cite{Kingman1982}. 
This process appears as the large population limit for a large class of reproduction models, including the Wright-Fisher \citep{Wright1938} and the Moran model \citep{Moran1958}.
The coalescent defines the genealogy of a sample and is meanwhile a common tool in population genetics.
We will use a time scaling to obtain Kingman's coalescent in the large population limit,
$N \to \infty$.
Thus, we assume that $\mu_i = \mu_i(N)$ for $i=1,2$ and $\nu = \nu(N)$, such that $\theta_i = \lim_{n\to\infty} 2\mu_i(N)N$ and $\rho = \lim_{n\to\infty} 2\nu_i(N)N$.

\begin{definition}[Kingman's coalescent]
 \label{Kingmancoal}
  The Kingman coalescent (or the $n$-coalescent) $(R_t)_{t \geq 0}$ is a continuous time Markov process with state space $\Pi_n$, 
the set of all partitions of $\{1,\dots,n\}$, and infinitesimal generator $Q = (q_{\xi\eta})_{\xi,\eta \in \Pi_n}$ given by:\\
\begin{equation}
  q_{\xi\eta} = \begin{cases}
		  - \frac{k(k-1)}{2} & \text{if } \xi = \eta\\
		  1 & \text{if } \xi \prec \eta\\ 
		  0 & otherwise
                \end{cases}
\end{equation}
where $k := |\xi|$ is the number of partition elements in $\xi$, and $\xi \prec \eta$ iff $\eta$ is obtained from 
$\xi$ by combining two partition elements of $\xi$.
The initial state $R_0 = \{ \{1\},\dots,\{n\} \}$ is the partition, where each $i \in \{1,..,n\}$ is its own partition element. 
\end{definition}

\begin{remark}
The partition $R_t = \{n_1,\dots,n_K\}$ contains the partition element $n_k$ with $i,j \in n_k \in R_t$ if and only if the $i$-th and the $j$-th individual of the sample have a common ancestor at time t.
Note that in Kingman's coalscent time is measured backwards. Furthermore the $n-$coalescent implicitly defines a random bifurcating tree with $n$ leaves. 
At time $t$ the tree has $k = |R_t|$ branches, where the $l-$th branch leads to the $i-$th leaf if $i \in n_l \in R_t$, see Figure \ref{gsfsmodel}.
\end{remark}

\begin{definition}[Kingman's coalescent]
  \label{Kingmancoaltree}
  We denote the random tree resulting from the above mechanism -- the
  Kingman coalescent -- by $\mathcal T$. We consider $\mathcal T$ as a
  partially ordered metric space with order relation $\preceq$ and
  metric $d_{\mathcal T}$ where the distance of two points in
  $\mathcal T$ is given by the sum of the times to their most recent
  common ancestor. We make the convention that $s\preceq t$ for
  $s,t\in\mathcal T$ if $s$ is an ancestor of $t$.
\end{definition}

In \citep{BaumdickerHessPfaffelhuber2010} we introduced the process $\mathcal G_t$, which describes the set of genes along the coalescent.
 In the same spirit we will define the process $\mathcal M_t$ with state space given by the set of simple counting measures on $[0,1]^2$.
 $\mathcal M_t$ accumulates mutations in $I\times J$ by adding the point $(u,dv)$ to $\mathcal M_t$ at rate $\mathcal G_t(u)dt \theta_2 dv$.
 The point $(u,v)$ corresponds to a mutation in gene $u$ at site $v$.

\begin{figure}
\begin{tikzpicture}[scale = 0.85]
% coalescent
\draw (2.125,5) -- (2.125,6);
\draw[dotted] (2.125,6) -- (2.125,6.5);
 \draw (0,3.8)--(0,0);
 \draw (0,3.8)--(1.5,3.8);
 \draw (0.75,3.8)--(0.75,5) -- (3.5,5) -- (3.5,1);
 \draw (3,0) -- (3,1)--(4,1) -- (4,0);
 \draw (1,0) -- (1,3.1)--(2,3.1) -- (2,0);
 \draw (1.5,3.1)--(1.5,3.8);
% gains
 \draw[font = \Large, color = purple] (2.125,5.31) node (gain1) {$\blacktriangledown$};
 \draw[color = purple,right] (gain1) node {\raisebox{-0.5cm}{$2$}};
 \draw[font = \Large, color = blue] (0.75,4.6) node (gain2) {$\blacktriangledown$};
 \draw[color = blue,right] (gain2) node {\raisebox{-0.5cm}{$1$}};
 \draw[font = \Large, color = green] (3.5,4.1) node (gain3) {$\blacktriangledown$};
 \draw[color = green,right] (gain3) node {\raisebox{-0.5cm}{$4$}};
 \draw[font = \Large, color = black] (0,0.7) node (gain4) {$\blacktriangledown$};
 \draw[color = black,right] (gain4) node {\raisebox{-0.5cm}{$3$}};
 %  losses
 \draw[font = \Large, color = purple] (1,1.65) node (loss1) {$\bullet$};
 \draw[color = purple,right] (loss1) node {\raisebox{-0.5cm}{$2$}};
%  \draw[font = \Large, color = blue] (1.5,3.4) node (loss2) {$\bullet$};
 \draw[font = \Large, color = green] (4,0.5) node (loss3) {$\bullet$};
  \draw[color = green,right] (loss3) node {\raisebox{-0.5cm}{$4$}};
%  mutations
 \draw[color = purple] (0.75,4.1) node (mut11) {\XSolid};
 \draw[color = purple,right] (mut11) node {\raisebox{-0.5cm}{$\,2$}};
 \draw[color = purple] (2,0.21) node (mut12) {\XSolid};
 \draw[color = purple,right] (mut12) node {\raisebox{-0.5cm}{$\,2$}};
 \draw[color = purple] (3.5,2.3) node (mut13) {\XSolid};
 \draw[color = purple,right] (mut13) node {\raisebox{-0.5cm}{$\,2$}};
 \draw[color = blue] (0,1) node (mut21) {\XSolid};
 \draw[color = blue,right] (mut21) node {\raisebox{-0.5cm}{$\,1$}};
 \draw[color = green] (4,0.9) node (mut31) {\XSolid};
 \draw[color = green,right] (mut31) node {\raisebox{-0.5cm}{$\,4$}};
 \draw[color = green] (3.5,3.3) node (mut32) {\XSolid};
 \draw[color = green,right] (mut32) node {\raisebox{-0.5cm}{$\,4$}};
 \draw[color = green] (3.5,1.95) node (mut33) {\XSolid};
 \draw[color = green,right] (mut33) node {\raisebox{-0.5cm}{$\,4$}};
% LEiste mit coalescent
\draw[right] node at (4.5,0.5){$\{\{1\},\{2\},\{3\},\{4\}\}$};
% \draw (4.8,1) -- (5.2,1);
\draw[right] node at (4.5,2.05){$\{\{1\},\{2\},\{3\},\{4,5\}\}$};
% \draw (4.8,3.1) -- (5.2,3.1);
\draw[right] node at (4.5,3.45){$\{\{1\},\{2,3\},\{4,5\}\}$};
% \draw (4.8,3.8) -- (5.2,3.8);
\draw[right] node at (4.5,4.4){$\{\{1,2,3\},\{4,5\}\}$};
% \draw (4.8,5) -- (5.2,5);
\draw[right] node at (4.5,6){$\{\{1,2,3,4,5\}\}$};
% subline
\draw (0,-0.2) node {1};\draw (1,-0.2) node {2};\draw (2,-0.2) node {3};\draw (3,-0.2) node {4};\draw (4,-0.2) node {5};
 \end{tikzpicture}
\begin{tabular}[b]{lcccc}
         &  {\color{blue}gene 1} &  {\color{purple}gene 2} &  {\color{black}gene 3} &  {\color{green}gene 4} \\
 refseq & {\color{blue} {\tt NNNNN }} & {\color{purple} {\tt NNNNN }} & {\color{black} {\tt NNNN }} & {\color{green} {\tt NNNNNNN }}  \\        
 ind. 1 & {\color{blue} {\tt --T-- }} & {\color{purple} {\tt -A--- }} & {\color{black} {\tt ---- }} & {\color{green} {\tt $\varnothing$ }}  \\
 ind. 2 & {\color{blue} {\tt ----- }} & {\color{purple} {\tt $\varnothing$ }} & {\color{black} {\tt $\varnothing$ }} & {\color{green} {\tt $\varnothing$ }}  \\
 ind. 3 & {\color{blue} {\tt ----- }} & {\color{purple} {\tt -AA-- }} & {\color{black} {\tt $\varnothing$ }} & {\color{green} {\tt $\varnothing$ }}  \\
 ind. 4 & {\color{blue} {\tt $\varnothing$ }} & {\color{purple} {\tt T---- }} & {\color{black} {\tt $\varnothing$ }} & {\color{green} {\tt ---AC-- }}  \\
 ind. 5 & {\color{blue} {\tt $\varnothing$ }} & {\color{purple} {\tt T---- }} & {\color{black} {\tt $\varnothing$ }} & {\color{green} {\tt $\varnothing$ }}  \\
 \end{tabular}
 \caption{\label{gsfsmodel}We model mutations within dispensable genes along a Kingman coalescent.
 Along each branch new genes can occur ($\blacktriangledown$), 
 and existing genes can get hit by single point mutations ({\tiny \XSolid }) or get lost by a gene loss event ($\bullet$).
 The left graph shows one realization of the process defined in Definition \ref{def:treeindexedMC}.
 The index at each symbol indicates the affected gene.
 In the middle the corresponding states of Kingman's coalscent are given. The table shows the resulting gene sequences.}
\end{figure}
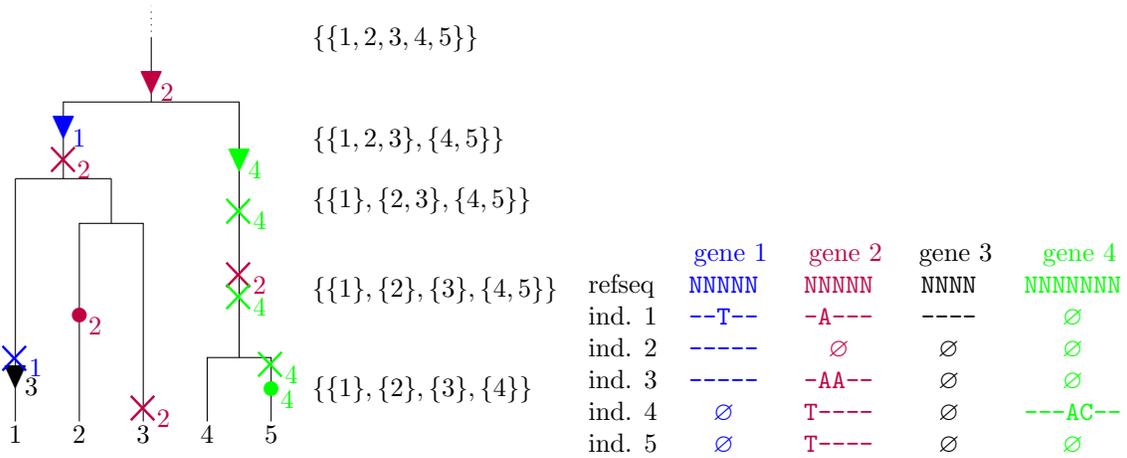

\begin{definition}[Tree-indexed Markov chain for gene gain, loss and site mutation]
\label{def:treeindexedMC}
  Let $I:=[0,1], J = (0,1]$ and let $\mathcal T$ be a Kingman coalescent, with an infinite lineage added at the root of the tree. Given
  $\mathcal T$, we define a Markov chain $\Gamma_{\mathcal T} =  (\mathcal M_t)_{t\in\mathcal T}$, indexed by $\mathcal T$, with
  state space $\mathcal N_f([0,1]^2)$, the space of finite counting measures on
  $[0,1]^2 = I \times (\{0\} \cup J)$.
%  The Markov property for the tree-indexed Markov chain  $\Gamma_{\mathcal T}$ states that for all $t\in\mathcal T$,  $(\mathcal M_s)_{t\preceq s}$ depends on $(\mathcal M_s)_{s\preceq    t}$ only through $\mathcal M_t$.  
Denoting by $\lambda_I$ the  Lebesgue measure on $I$, $\Gamma_{\mathcal T}$ makes transitions forwards in time, i.e. from the root to the leaves,
  \begin{equation}
    \begin{aligned}
      &\text{from }m \text{ to } m + \delta_{(u,0)} \quad\ \text{ at rate }
      \tfrac{\theta_1}{2}
      \lambda_I(du),\\
      &\text{from }m \text{ to } m - m|_{\{u\}\times I}  \text{ at rate }
      \tfrac\rho 2 m(u,0),\text{ and}\\
      &\text{from }m \text{ to } m + \delta_{(u,v)} \quad\ \text{ at rate }
      \tfrac{\theta_2}{2} m(u,0)\lambda_I(dv)
    \end{aligned}
  \end{equation}
  along $\mathcal T$. Taking into account that the tree $\mathcal T$
  has $n$ leaves, one for each individual of the sample, we denote
  these leaves by $1,\dots,n\in\mathcal T$. In this setting, $\mathcal
  M_1,\dots,\mathcal M_n$ describe the genes and the mutations within these genes present in individuals
  $1,\dots,n$.
\end{definition}

A graphical illustration of the Markov chain along the coalescent is given in Figure \ref{gsfsmodel}.
Note that all gained points in $[0,1]^2$ are almost surely different, so $\mathcal M_i$ is a.\,s.\ a simple counting measure.
Thus we identify counting measures with their support in our notation, i.e.\ if $g \in \mathcal N_f(I \times J)$ has no double points, there exist pairwise different
$(u_1,v_1)\dots(u_m,v_m) $ with $g = \sum_{i=1}^{m} \delta_{(u_i,v_i)}$. In this case we will also write $g = \{(u_1,v_1)\dots(u_m,v_m)\}$.\\
Since $v=0$ is almost surely never hit by a mutation for any gene, we may use $[0,1] \times \{ 0\}$
to model the presence and absence of genes. For simplicity we will assume that mutations are uniformly chosen from $J = (0,1]$.\\
Setting $\mathcal G_i(u) := \mathcal M_i(u,0)$, we will regain the corresponding tree indexed Markov chain of the infinitely many genes model, as given in Def 2.2 in \citep{BaumdickerHessPfaffelhuber2010}.

\section{Results}
\label{gsfresults}

\subsection{Joint gene and site frequency spectrum}
\label{S:35}

In contrast to the approach taken here, the expected frequency of mutated sites has so far
only been investigated for genes that can never get lost or gained.
Such genes we will call \emph{essential core genes}. 
These essential core genes are assumed to be present in any individual at any time. In our setting of Definition \ref{def:treeindexedMC}
an essential core gene $u$ has been gained at $-\infty$ and accumulates site mutations along the tree $\mathcal T$ at rate $\tfrac{\theta_2}{2}$, while it can not get lost.
% folgenden Satz hab ich hrausgenommen da er egal ist für dieses Paper
% By definition, in the infinitely many genes model a fixed set of core genes are present in all individuals at any time.
For these essential core genes the classical results for the site frequency spectrum hold. 
Moments of the site frequency spectrum were computed e.g. in \citep{Fu1995}.

First we recall the results for the classical site frequency spectrum.
Throughout this Section we will fix the size of the sample to $n$ individuals.

\begin{definition}[Site frequency spectrum for an essential core gene]
\label{sfsforessential}
 Consider the gene sequence of a essential core gene $u$,
i.e. $u$ is a single gene, which cannot get lost.
In the setting of Definition~\ref{def:treeindexedMC} this corresponds to
the Markov chain  $(\mathcal M_t)_{t \in \mathcal T}$,
if $r$ is the root of $\mathcal T$,
given that $\mathcal M_r = \delta_{(u,0)}$ with $\rho = 0$.
%gene $u$ is never lost ($\rho=0$), has been gained at $-\infty$ 
 Now site mutations within $u$ occur at rate $\tfrac {\theta_2}{2}$ along the lineages of $\mathcal T$.
 For $J = (0,1]$, let $\mathcal C^{u}_t := \mathcal M_t|_{u \times J} \in \mathcal N_f(J) $ be the finite counting 
measure describing the site mutations of the essential core gene $u$ along the coalescent $\mathcal T$.
 So $\mathcal C^{u}_i$ is the set of mutated sites within gene $u$ present in individual $i$.
 For a sample of size $n$ let $C^{u}_s$ be the number of sites where $s$ individuals in the sample carry the same mutation.
 I.e.\ 
 $$ C^u_s = \big |  \big\{ v \in J: v \in \mathcal C^u_i \text{ for exactly } s \text{ different } i \in \{1,\dots,n\}           \big\}   \big |.$$
  Then $C^{u}_s$ for $ 1 \leq s < n$ is called the site frequency spectrum of the essential core gene $u$.
\end{definition}

\begin{theorem}[Classical site frequency spectrum]
 \label{sitefreqspec}
Consider the setting of Definition \ref{sfsforessential}.
Suppose the site mutation rate is given by $\tfrac{\theta_2}{2}$.
% and we have a sample of size $n$.
The expected site frequency spectrum for an essential core gene $u$ is given by
\begin{equation}
\label{sfsformula}
 \mathbb E [C^u_s] = \frac{\theta_2}{s}, \qquad s=1,\dots,n-1.
\end{equation}
\end{theorem}

A similar formula can be computed for the frequencies of dispensable gene sequences themselves.
Therefore ignore site mutations for once and consider only the presence and absence of the genes.
To model the gene frequencies the infinitely many genes model was introduced
in \cite{BaumdickerHessPfaffelhuber2010}.
As a matter of fact, the infinitely many genes model was developed in spirit of the infinitely many sites model.
The main difference between the infinitely many sites model and the infinitely many genes model is that genes are allowed to get lost again, 
while each mutation, once it arose, will be present in all offspring. Thus the gene frequency spectrum differs from the site frequency spectrum.
The following result for the gene frequency spectrum holds.

\begin{definition}
If we set $\mathcal G_i(u) := \mathcal M_i(u,0)$, the \emph{gene frequency spectrum (of the dispensable genome)} is
given by $G_1,\dots,G_n$, where
\begin{align}\label{eq:Gk}
  G_k := |\{u\in I: u\in\mathcal G_i \text{ for exactly } k \text{
    different }i\}|.
\end{align}
\end{definition}

\begin{theorem}[gene frequency spectrum]\label{genefreqspec}
  For $G_1,\dots,G_n$ as above,
  \begin{equation*}
    \begin{aligned}
      \mathbb E[G_k] & = \frac{\theta}{k}\frac{(n-k+1)\cdots
        n}{(n-k+\rho)\cdots (n-1+\rho)} \qquad \qquad k=1,\dots,n.
    \end{aligned}
  \end{equation*}
\end{theorem}

\begin{proof}
see \cite{BaumdickerHessPfaffelhuber2010}
\end{proof}

Theorem~\ref{genefreqspec} describes the frequency of dispensable
genes, which can get lost.
Theorem~\ref{sitefreqspec} describes the site frequency spectrum only within genes that can never be lost.
In contrast to essential core genes, dispensable genes can be present at any frequency such that
equation~\eqref{sfsformula} no longer holds for a dispensable gene sequence.
% but the site frequency spectrum for dispensable genes differs from \eqref{sfsformula}.
% Theorem \ref{sitefreqspec} does only hold for sites within essential core genes and not for the sequences of dispensable genes.
To address this issue we will compute the expected joint gene and site frequency spectrum. In Corollary~\ref{condsitefreqspec} this enables us to obtain an analogous result to Theorem~\ref{sitefreqspec} for dispensable genes.

\begin{definition}
\label{def:jointfs}
The \emph{joint gene and site frequency spectrum (for the dispensable genome)} is
given by $G_{1,1},\dots,G_{1,n},G_{2,1},\dots,G_{2,n},\dots,G_{n,n}$, where
\begin{align*}\label{eq:Gk}
   G_{k,s} := \big|\big\{ (u,v) \in I \times I : u \in \mathcal G_i \text{ for exactly } k \text{ different }i \text{ ,namely } i_1,\dots,i_k \text{, and }  \\
   (u,v) \in \mathcal M_{i_j} \text{ for exactly } s \text{ different } i_j \text{ with } j \in \{1,\dots,k\}   \big\}   \big |
\end{align*}
\end{definition}

\begin{theorem}[Joint gene and site frequency spectrum]
\label{jointfs}
Suppose the gene gain rate is given by $\tfrac{\theta_1}{2}$ and the gene loss rate is given by $\tfrac{\rho}{2}$
Suppose further the mutation rate is given by $\tfrac{\theta_2}{2}$ and the sample size is $n$. 
  For $G_{k,s}$ as above and $s < k$,
  \begin{equation*}
    \begin{aligned}
      \mathbb E[G_{k,s}] & = %???
\frac{\theta_1}{k} \frac{(n-k+1)\cdots n}{(n-k+\rho)\cdots(n-1+\rho)} \frac{\theta_2}{s} \frac{k}{n} \binom{n-1}{s}^{-1} \sum_{j=0}^{n-s-1} \frac{j+1}{j+1+\rho} \binom{n-j-2}{s-1}\\
\text {and for } k=s \\
      \mathbb E[G_{k,s}] & = %???
\frac{\theta_1}{k} \frac{(n-k+1)\cdots n}{(n-k+\rho)\cdots(n-1+\rho)} \frac{\theta_2}{s} sk 
       \sum_{j=1}^{n-s+1} \frac{1}{j(j-1+\rho)} \binom{n}{j}^{-1} \binom{n-k}{j-1}.
    \end{aligned}
  \end{equation*}
%ALTES
%   \begin{equation*}
%     \begin{aligned}
%       \mathbb E[G_{k,s}] & = %???
% \frac{\theta_1}{k} \frac{k!}{(n-k+\rho)\cdots(n-1+\rho)} \frac{\theta_2}{s} \binom{n}{k} \frac{sk}{(k-s)(k-s+1)}  \\
%     & \quad   \sum_{j=1}^{n-s+1} \frac{1}{j(j-1+\rho)} \binom{n}{j}^{-1} \sum_{m=1}^{k-s+1} m (m-1) \binom{k-s+1}{m} \binom{n-k}{j-m},\\
% \text {and for } k=s \\
%       \mathbb E[G_{k,s}] & = %???
% \frac{\theta_1}{k} \frac{k!}{(n-k+\rho)\cdots(n-1+\rho)} \frac{\theta_2}{s} \binom{n}{k}  sk \\
%     & \quad   \sum_{j=1}^{n-s+1} \frac{1}{j(j-1+\rho)} \binom{n}{j}^{-1} \binom{n-k}{j-1}.
%     \end{aligned}
%   \end{equation*}
\end{theorem}

\subsection{Site frequency spectra for dispensable genes}

Let us define the site frequency spectrum for a dispensable gene
present in a given frequency~$k$.
\begin{definition}[Site frequency spectrum for dispensable genes]
\label{def:sfsdispensablegenes}
Given $\mathcal M_i$ from Definition \ref{def:treeindexedMC}, define $\mathcal S^u_i(.) := \mathcal M_i(u,.)$ for $i = 1,\dots,n$.
Let $u \in [0,1]$ be a gene. Then the site frequency spectrum of gene $u$ is given by $S^u_1,\dots,S^u_{F(u)}$, where
\begin{equation*}
 S^u_s := |\{v \in J : v \in \mathcal S^u_i \text{ for exactly $s$ different $i$ with $\mathcal M_i(u,0) = 1$}   \}|
\end{equation*}
and $F(u)$ is the frequency of gene $u$, i.e.\ 
$F(u) := |\{i\in \{1,\dots,n\}: \mathcal M_i(u,0) = 1 \}| $.
\end{definition}

From Theorem \ref{jointfs} we can now derive the expected site frequency spectrum for a gene in frequency~$k$.

\begin{corollary}[Conditional site frequency spectrum]
\label{condsitefreqspec}
The site frequency spectrum in genes present in exactly $k$ out of $n$ individuals is given for $s<k$ by
 \begin{equation}
  \label{EkGs}
\mathbb E[S^u_s \ | F(u) = k ] = \frac{\theta_2}{s} \frac{k}{n} \binom{n-1}{s}^{-1} \sum_{j=0}^{n-s-1} \frac{j+1}{j+1+\rho} \binom{n-j-2}{s-1}\\
%   \mathbb E[S^u_s \ | F(u) = k ] = \frac{\theta_2}{s} \frac{sk}{(k-s)(k-s+1)} \sum_{m=1}^{k-s+1} \sum_{j=m}^{n-k+m}  \binom{n}{j}^{-1} \frac{m (m-1)}{j(j-1+\rho)}    \binom{k-s+1}{m}   \binom{n-k}{j-m}
 \end{equation}
and 
\begin{equation}
 \mathbb E[S^u_k\ | F(u) = k] = \frac{\theta_2}{k} k^2 \sum_{j=1}^{n-k+1}  \binom{n}{j}^{-1} \frac{1}{j(j-1+\rho)}     \binom{n-k}{j-1} \label{freqk}
\end{equation}

In particular, for $k = n > s$
\begin{equation}
  \label{EnGs}
    \mathbb E[S^u_s\ |{F(u) = n}] = \frac{\theta_2}{s} \binom{n-1}{s}^{-1} \sum_{j=0}^{n-s-1} \frac{j+1}{j+1+\rho} \binom{n-j-2}{s-1}\\
%ALTES   \mathbb E[S^u_s\ |{F(u) = n}] = \frac{\theta_2}{s} \binom{n-1}{s}^{-1} \sum_{m=2}^{n-s+1} \binom{n-m}{s-1} \frac{m-1}{m-1+\rho}.
 \end{equation}
and
\begin{equation}
 \mathbb E[S^u_n\ |{F(u) = n}] = \frac{\theta_2}{\rho} \label{freqn}
\end{equation}

\end{corollary}

Figure \ref{sfsbunt} shows the site frequency spectrum for dispensable genes in frequency $k=5$ and $k=10$ for $n=10$.

\begin{remark}
 For $\rho > 0$ the site frequency spectrum for dispensable genes differs from the classical result,
 even for dispensable genes which are present in all individuals of the sample.
While this seems counterintuitive at the first glimpse, it is easily explained.
Unlike essential core genes, it is in principle possible for a dispensable gene to get lost. Therefore, a dispensable gene, which is present in all individuals, is more likely to occur
if the underlying coalescent has small branch lengths.
%, than in a coalescent with long branches.
As mutations are gained along the branches of the coalescent, dispensable genes in frequency $n$ will on average carry less mutations than essential core genes.
Of course, this is only the case if two independent genes from different populations are compared.
In this case the underlying coalescent of the dispensable gene does not depend on the coalescent of the essential core gene.
%If two genes of different frequencies are compared within the same population
%Two genes within the same population will get less correlated the higher the recombination rate between genes or the horizontal gene transfer rate is.
\end{remark}

\begin{remark}
 Note that, for $\rho \to 0$ and $k=n$, equation \eqref{EnGs} converges to the classical site frequency spectrum $\frac{\theta_2}{s}$.
 In contrast, if $k$ is smaller than $n$, equation \eqref{EkGs} does not converge to the classical site frequency spectrum.
 This can be easily seen from coalescent theory.
 Genes present in all individuals are gained at some time before the MRCA, while genes in frequency smaller than $n$ are gained along branches within Kingman's coalescent.
 A subtree with $k$ leaves within a $n$-coalescent is only a $k$-coalescent, if the $k$ leaves are chosen randomly from the $n$ available leaves.
 This is clearly not the case, if the $k$ leaves are chosen such that they form a nested cluster,
 a situation which appears when considering leaves carrying the same gene and small small loss rates.
\end{remark}

\begin{remark}
 In equation \eqref{EnGs} we can derive the same result for $s=1$ by a simple approach.
Therefore start a Kingman coalescent with $n$ lineages, where each line gets lost at rate $\frac{\rho}{2}$.
In Section \ref{proofs} we will have a closer look at this lineage loosing coalescent.
 For $m=2,\dots,n$ add up the expected times the lineage loosing coalescent has $m$ lineages, $\mathbb E[T_m] = \tfrac{2}{m(m-1)+m\rho}$, multiplied
by the expected number of external branches in a $n$-coalescent, when there are $m$ lineages left. The former quantity is given by $\tfrac{m(m-1)}{n-1}$, see \cite{Janson2011}. 
\end{remark}

\begin{remark}
 In Theorem \ref{jointfs} as well as in \eqref{freqk} and \eqref{freqn} in Corollary \ref{condsitefreqspec} we have given formulas 
for the expected number of mutated sites present in $k$ out of $k$ individuals.
 The corresponding size in the classical site frequency spectrum of an essential core gene would be infinitely large,
 as the essential gene accumulated site mutations for an infinite time before the MRCA of the $k$ individuals.
 In order to identify the mutated sites present in all individuals one would need a reference sequence, 
 like the ancestral sequence at the time the gene was introduced into the population. As this sequence is most likely not available in practice
 normally only site mutations in frequency $s<k$ are considered.
\end{remark}

\section{The effect on estimates}
\label{effectonestimates}

\begin{figure}
\begin{overpic}[width=0.33\textwidth]{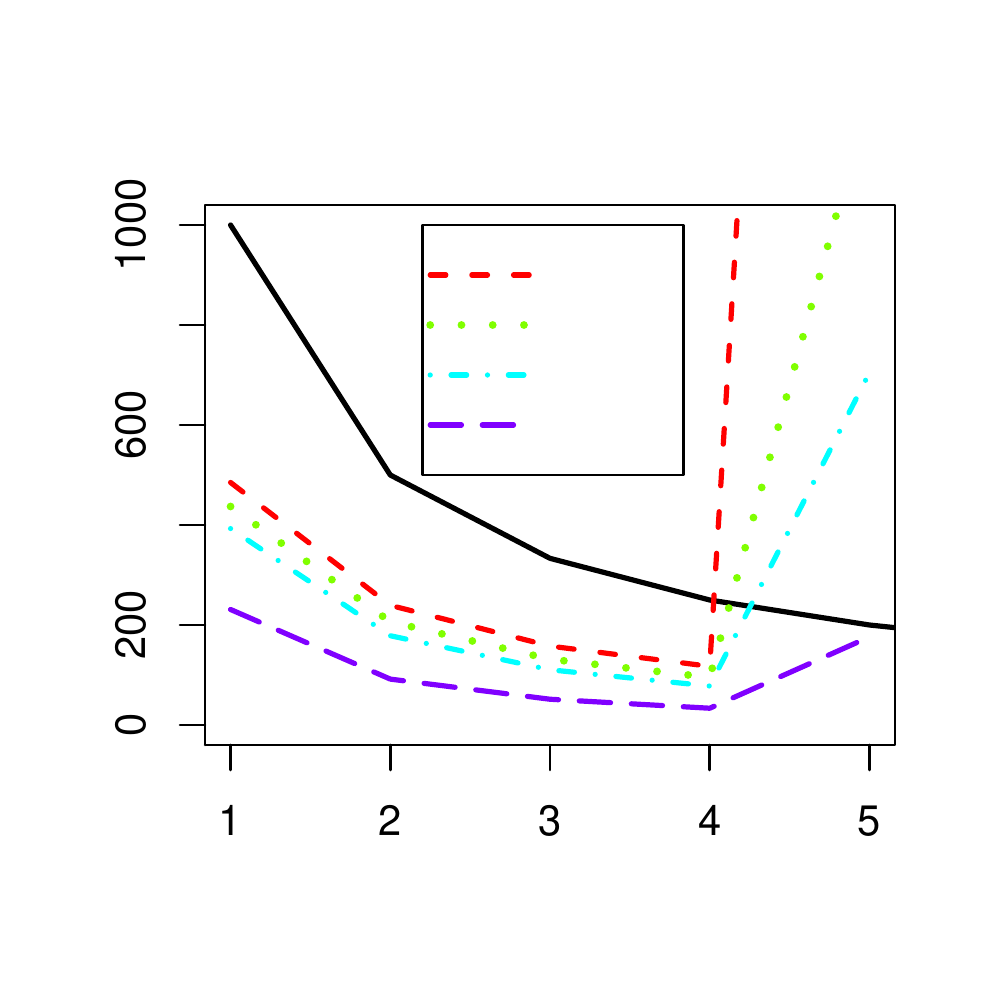}
\put(25,8){\footnotesize$s$ -- frequency of SNP}
\put(0,32){\rotatebox{90}{\footnotesize$\mathbb E[S^u_s \ | F(u) = k ]$}}
%legend
\put(54,72){\tiny$\rho\!=\!0.1$}
\put(54,67){\tiny$\rho\!=\!0.5$}
\put(54,62){\tiny$\rho\!=\!1$}
\put(54,57){\tiny$\rho\!=\!5$}
\end{overpic}
\begin{overpic}[width=0.66\textwidth]{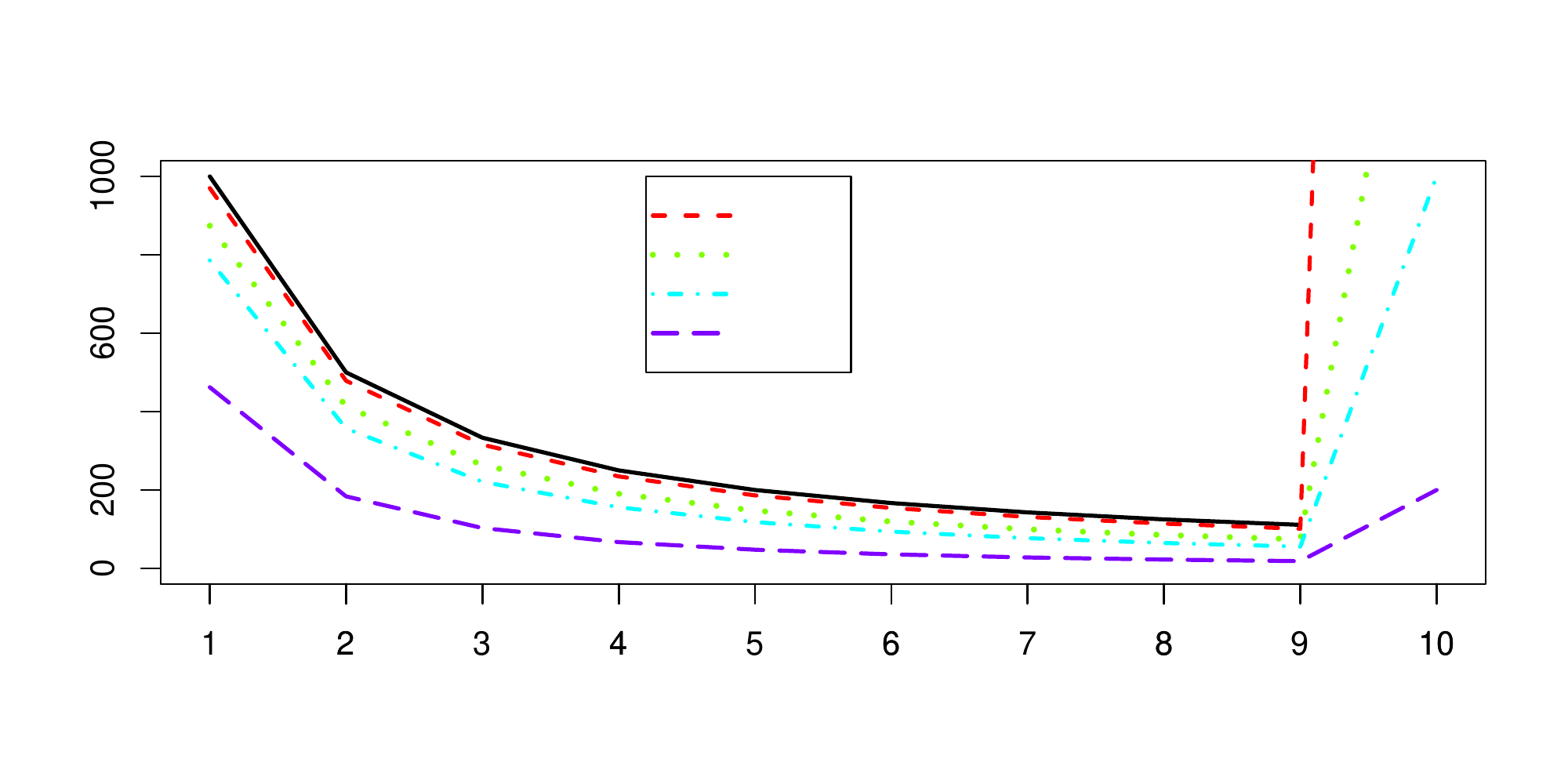}
\put(36,4){\footnotesize$s$ -- frequency of SNP}
\put(0,16){\rotatebox{90}{\footnotesize$\mathbb E[S^u_s \ | F(u) = k ]$}}
%legend
\put(47,36){\tiny$\rho\!=\!0.1$}
\put(47,33.5){\tiny$\rho\!=\!0.5$}
\put(47,31){\tiny$\rho\!=\!1$}
\put(47,28.5){\tiny$\rho\!=\!5$}
\end{overpic}
\caption{The graph shows the classical site frequency spectrum (solid black)
and the site frequency spectrum for dispensable genes for $\rho$ equal to $0.1$ (-\,-\,-), $0.5$ ($\cdots$), $1$ ($\cdot$\,-\,$\cdot$\,-) and $5$ (\textendash\ \textendash).
The right graph shows the frequency spectrum for genes present in all $n=10$ individuals, while the left graph shows the spectrum for 
genes present in 5 out of 10 individuals.
The mutation rate is fixed at $\theta_2 =1000$.
The site frequency spectrum for dispensable genes present in all individuals differs from the classical site frequency spectrum.
The higher the probabilty to loose a gene, the less mutated sites are expected. 
If $\rho$ gets small the site frequency spectrum converges to the classical site frequency spectrum only if $k=n$ (right graph). \label{sfsbunt}}
\end{figure}
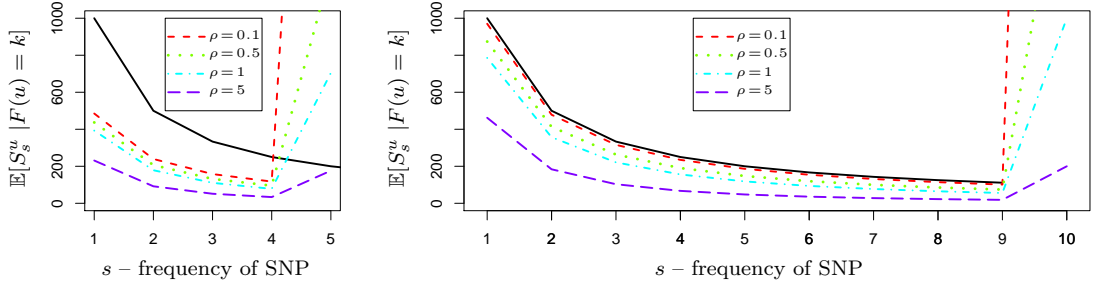

In this section we will highlight the consequences of Corollary \ref{condsitefreqspec} for estimators based on the site frequency spectrum.
It is possible to define many different estimators for the scaled site mutation rate $\theta_2$ based on the observed site frequencies.
Here we will focus on two commonly used estimators, and bear in mind that similar results will hold for any frequency based estimator.
We first introduce the estimators for the sequence of an essential core gene, and thereafter turn to the dispensable gene sequence.

A frequently used estimate for $\theta_2$ is given by Watterson's estimator \citep{Watterson1975},  which is defined by
\begin{equation}
 \widehat \theta_{W,k} := \frac{S_{\text{seg}}}{\sum_{s=1}^{k-1} \frac{1}{s}}.
\end{equation}
Here $S_{\text{seg}}$ is the total number of segregating sites observed in the sample of size $k$.
For an essential core gene $u$, we have $S_{\text{seg}} =  \sum_{s=1}^{k-1}  C_s^u $ and so the expected number of segregating sites is given by
\begin{equation}
      \mathbb E[S_{\text{seg}}] = \sum_{s=1}^{k-1} \mathbb E[ C_s^u ] = \sum_{s=1}^{k-1}  \frac{\theta_2}{s}
\end{equation}
and thus $\mathbb E[\widehat \theta_{W,k} ] = \theta_2$.

Another famous estimator for $\theta_2$ is given by $\widehat \pi$, sometimes referred to as Tajima's estimator \citep{Tajima1983}.
The estimator $\widehat \pi$ equals the sum of pairwise observed SNPs defined by
\begin{equation}
 \widehat \pi_k := \binom{k}{2}^{-1} \sum_{i<j}^k \pi_{ij},
\end{equation}
with $\pi_{ij}$, the number of sites which differ between individual $i$ and individual $j$.
Given an essential core gene $u$ the expectation of $\widehat \pi$ is again given by
\begin{align*}
 \mathbb E[\widehat \pi_k] = \sum_{s=1}^{k-1} \mathbb E[C_s^u] \frac{s(k-s)}{\binom{k}{2}} = \theta_2 \frac{2}{k(k-1)} \sum_{s=1}^{k-1} (k-s) = \theta_2.
\end{align*}

In contrast we have the following result if, instead of an essential core gene, a dispensable genes $u$ is given.

\begin{theorem}[Estimating the site mutation rate in a dispensable gene]
\label{dispestimates}
Consider a dispensable gene $u$, which appears in $k$ out of $n$ individuals within the sample.
Given $F(u)=k$ the site frequency spectrum $S_1^u,\dots,S_{k-1}^u$ is given by Definition \ref{def:sfsdispensablegenes}
and we set 
\begin{equation*}
  S_{\text{seg}} := \sum_{s=1}^{k-1} S_s^u \\
\end{equation*}
The expected number of segregating sites in a dispensable gene is now given by
\begin{align*}
  \mathbb E[S_{\text{seg}}] &= \sum_{s=1}^{k-1} \mathbb E[ S_s^u | F(u) = k  ] \\
%   &= \sum_{s=1}^{k-1}  \frac{\theta_2}{s} \frac{sk}{(k-s)(k-s+1)} \sum_{m=1}^{k-s+1} \sum_{j=1}^{n-k+m}  \binom{n}{j}^{-1} \frac{m (m-1)}{j(j-1+\rho)}    \binom{k-s+1}{m}   \binom{n-k}{j-m}\\
%   &= \theta_2 k \sum_{m=2}^{k} \sum_{j=1}^{n-k+m}  \binom{n}{j}^{-1} \frac{m (m-1)}{j(j-1+\rho)}  \binom{n-k}{j-m} \sum_{s=m}^{k}   \binom{s}{m} \frac{1}{(s-1)s} \\
% ALTES RESULTAT:  &= \theta_2  \sum_{m=2}^{k} \sum_{j=m}^{n-k+m}  \binom{n}{j}^{-1} \frac{m}{j(j-1+\rho)}  \binom{n-k}{j-m} \binom{k}{m}.
  &= \theta_2 \frac{k}{n} \sum_{s=1}^{k-1} \frac{1}{s} \binom{n-1}{s}^{-1} \sum_{j=0}^{n-s-1} \frac{j+1}{j+1+\rho} \binom{n-j-2}{s-1}
\end{align*}

Consequently, the expected value of Watterson's estimator for a dispensable gene $u$ in $k$ out of $n$ individuals is given by
\begin{equation}
 \label{eq:wattersondisp}
 \mathbb E[\widehat \theta_{W,k}] = \theta_2 \frac{\frac{k}{n} \sum_{s=1}^{k-1} \frac{1}{s} \binom{n-1}{s}^{-1} \sum_{j=0}^{n-s-1} \frac{j+1}{j+1+\rho} \binom{n-j-2}{s-1}}{\sum_{s=1}^{k-1} \frac{1}{s}} \leq \frac{k}{n} \theta_2 .
\end{equation}

For Tajimas estimator $\widehat \pi_k := \sum_{i<j}^k \pi_{ij}$, numbering the individuals which carry gene $u$ by $1,\dots,k$, the expectation is given by
\begin{align}
\mathbb E[\widehat \pi_k] &=  \sum_{s=1}^{k-1} \mathbb E[S_s^u| F(u) = k] \frac{s(k-s)}{\binom{k}{2}}  \notag \\
% 	&= \theta_2 \frac{2}{k(k-1)} \sum_{s=1}^{k-1} \frac{sk}{(k-s+1)} \sum_{m=1}^{k-s+1} \sum_{j=1}^{n-k+m}  \binom{n}{j}^{-1} \frac{m (m-1)}{j(j-1+\rho)}    \binom{k-s+1}{m}   \binom{n-k}{j-m} \\
%         &= \theta_2  \frac{2}{k(k-1)} k \sum_{m=2}^{k} \sum_{j=1}^{n-k+m}  \binom{n}{j}^{-1} \frac{m (m-1)}{j(j-1+\rho)}  \binom{n-k}{j-m} \sum_{s=m}^{k}   \binom{s}{m} \frac{k-s+1}{s} \\
%         &= \theta_2  \frac{2}{k-1} \sum_{m=2}^{k} \sum_{j=1}^{n-k+m}  \binom{n}{j}^{-1} \frac{m (m-1)}{j(j-1+\rho)}  \binom{n-k}{j-m} \binom{k}{m} \frac{1+k}{m(m+1)} \\
%  ALTES RESULTAT        &= \theta_2  \frac{2}{k-1} \sum_{m=2}^{k} \sum_{j=m}^{n-k+m}  \binom{n}{j}^{-1} \frac{(m-1)}{(m+1)}\frac{(1+k)}{j(j-1+\rho)}  \binom{n-k}{j-m} \binom{k}{m} \leq \frac{k}{n}\theta_2 \label{eq:tajimadisp}
        &= \theta_2  \frac{2}{k(k-1)} \frac{k}{n} \sum_{s=1}^{k-1} (k-s)  \binom{n-1}{s}^{-1} \sum_{j=0}^{n-s-1}   \frac{j+1}{j+1+\rho} \binom{n-j-2}{s-1} \leq \frac{k}{n}\theta_2 \label{eq:tajimadisp}
\end{align}

\end{theorem}

In particular, for a dispensable gene $u$ in frequency $k<n$, Watterson's estimator as well as Tajima's estimator will both underestimate the scaled site mutation rate $\theta_2$.

In the Theorem above we set $\rho = 0$ to get an upper limit for the expectation.
In fact, if $\rho = 0$, we get 
\begin{equation*}
\mathbb E[S_s^u| F(u) = k] =  \frac{k}{n} \mathbb E[C_s^u],
\end{equation*}
such that in \eqref{eq:wattersondisp} and \eqref{eq:tajimadisp} equality holds.

While for larger $k$ the bias of the estimators gets smaller, the difference between 
$\widehat \theta_W $ and $\widehat \pi$ increases if $\rho > 0$.
This effect is best illustrated setting the gene frequency to $k=2$ or $k=n$ in Theorem~\ref{dispestimates}.

\begin{corollary}
 Setting $k = n$ in Theorem \ref{dispestimates} gives
\begin{align*}
 \mathbb E[\widehat \theta_{W,n}] &= \theta_2 \frac{\sum_{s=1}^{n-1} \frac{1}{s+\rho}}{\sum_{s=1}^{n-1} \frac{1}{s}}\\
 \mathbb E[\widehat \pi_n] &=  \theta_2 \left( 1- 2 \rho  \frac{n+1}{n-1} \sum_{j=0}^{n-2} \frac{1}{(j+1+\rho)(j+2)(j+3)} \right) 
\end{align*}
 While setting $k = 2$ results in
 \begin{align*}
 \mathbb E[\widehat \theta_{W,2}] = \mathbb E[\widehat \pi_2] =\theta_2 \binom{n}{2} \sum_{j=0}^{n-2} \frac{j+1}{j+1+\rho}
\end{align*}
\end{corollary}

Figure \ref{estimatefigs} illustrates the results of Theorem \ref{dispestimates}.

\begin{figure}
 \begin{center}
 \begin{overpic}[width = 0.7\textwidth,trim=0 30 0 50,clip]{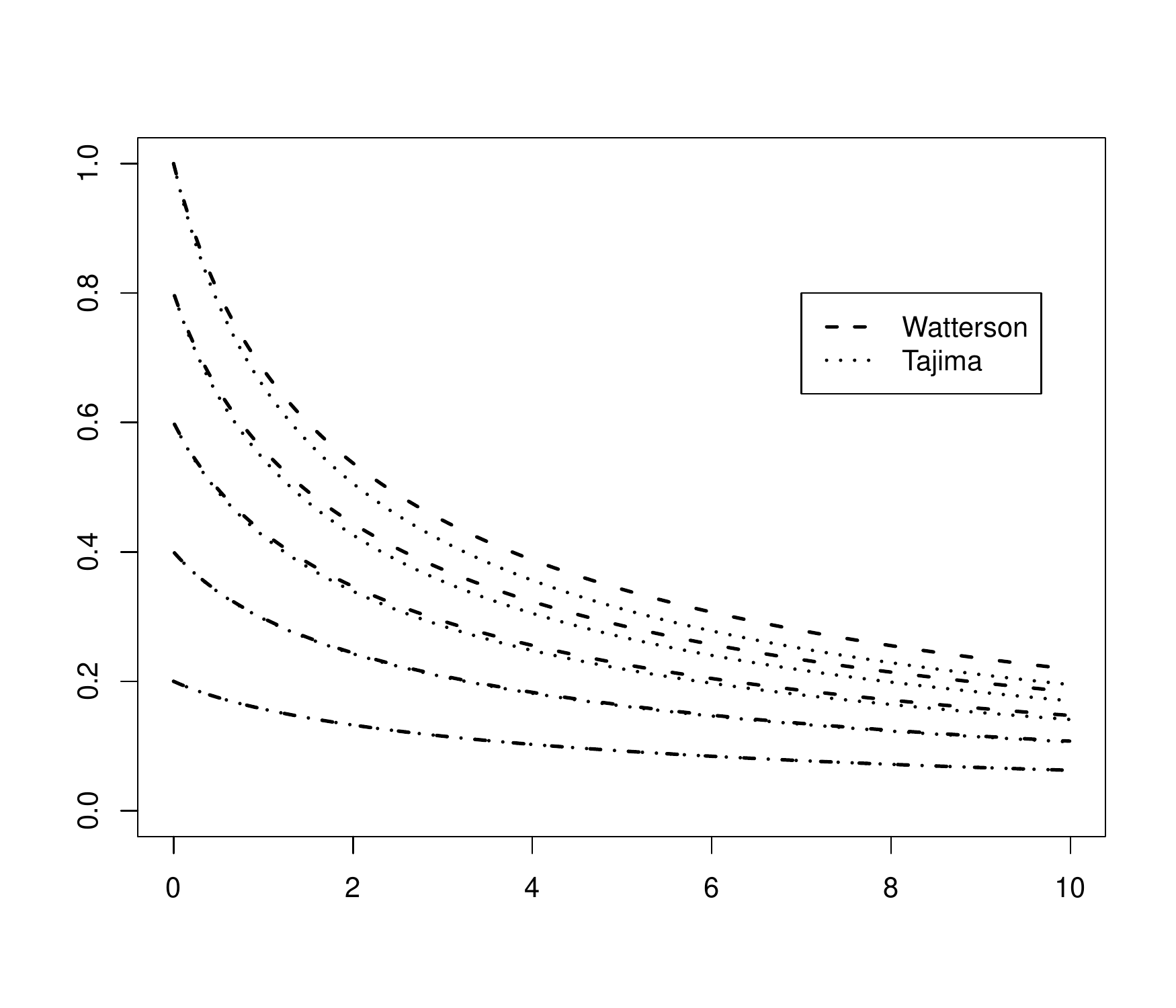}
  \put(52,0) {$\rho$}
%  \put(0,55) {$\frac{\mathbb E[\widehat \theta_{W,k}]}{\theta_2}$ }
 \put(0,25) {\rotatebox{90}{$\mathbb E[\widehat \theta_{W,k}] / \theta_2$ and $\mathbb E[\widehat \pi_k] / \theta_2$ }}
%  \put(-10,45) {$\mathbb E[\widehat \pi_k] / \theta_2$ }
 \end{overpic}
 \end{center}
\caption{The effect of $\rho$ and $k$ on the two estimators is shown. The solid line shows Wattersons estimator for $k \in \{2,4,6,8,10\}$,
 while the dashed line shows Tajimas estimator. The lowest line is for $k=2$, the uppermost line shows the estimators for $k=n=10$.
 The sample size is fixed at $n=10$. \label{estimatefigs}}
\end{figure}

\section{Discussion}
\label{discussion}
%impact on estimates 
The $k$ individuals which possess a dispensable gene span a subtrees of Kingman's $n-$coalescent, which is nested if $\rho \to 0$. This subtrees differs from Kingman's $k$-coalescent as the $k$ individuals are not chosen independently.
The subtree has lower depth than a coalescent, and the relative branch lengths are changed. This difference causes the distortion of the estimates.
There is a monotonic relationship between the frequency of a dispensable gene (the gene loss rate $\rho$) and the estimate of $\theta_2$.
The lower (larger) the frequency of a dispensable gene (the gene loss rate $\rho$) is, the larger is the bias for the estimate of the scaled site mutation parameter $\theta_2$.
% The estimate of $\theta_2$ is less biased the larger the frequency $k$ of the dispensable gene is. 

Although $\widehat \theta_W$ and $\widehat \pi$ are both negatively biased the relation $\mathbb E[\widehat \theta_W] = \mathbb E[\widehat \theta_W]$ does no longer hold for $\rho > 0$ and $k > 2$.
In contrast to the bias the difference between the two estimators $\widehat \pi - \widehat \theta_W $ gets larger the larger $k$ is and shows a non-monotonic behavior for $\rho$.
% For low values of $\rho$ the effect is too small to cause a significant difference. For high values of $\rho$ the  ...
%impact on Tajimas D
The normalized difference between $\widehat \theta_W$ and $\widehat \pi$ is well known as Tajima's D \citep{Tajima1989a}.
%say something about TajiD
Tajima's D is capable of detecting non-neutral evolving sequences. For example, negative values hint at purifying selection, while positive values suggest balancing selection.
As $\mathbb E[\widehat \theta_W] \neq \mathbb E[\widehat \pi]$, Tajima's D is biased for dispensable genes.
Our results show that in dispensable genes an excess of singleton and low frequency site mutations should be expected. 
Figure \ref{sfsnorm} illustrates the excess of singleton site mutations for different gene loss rates.
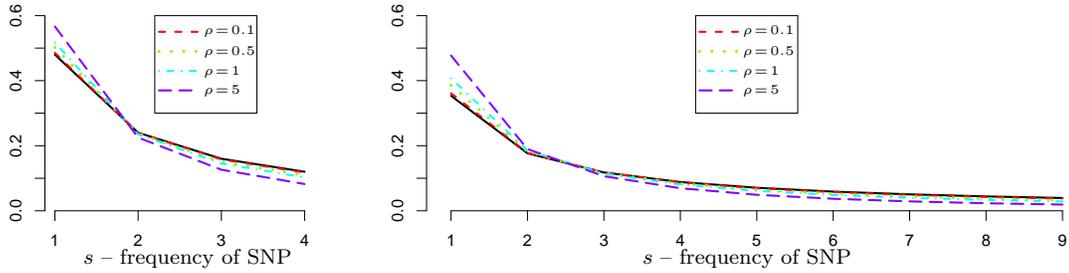
\begin{figure}
% \fbox{
\begin{overpic}[width=0.33\textwidth,trim=0 40 15 40,clip]{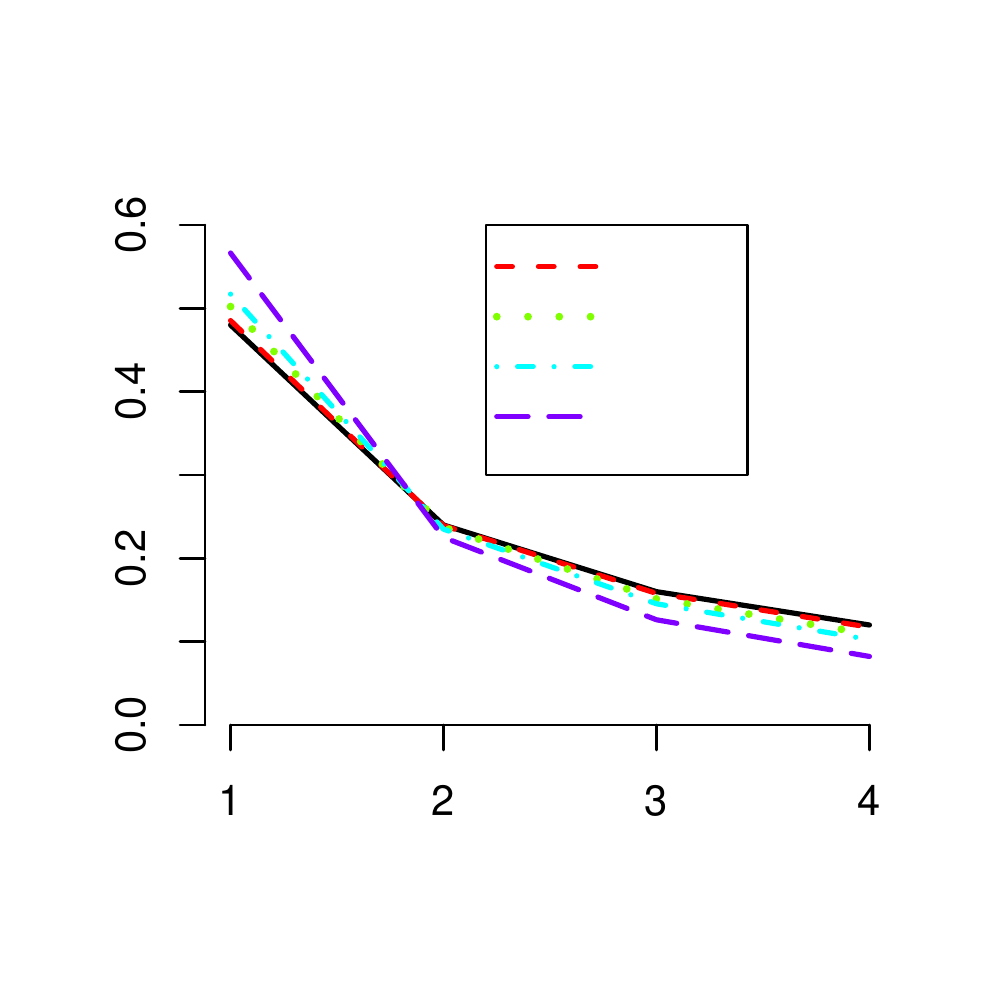}
\put(32,0){\footnotesize$s$ -- frequency of SNP}
%legend
\put(64.5,62){\tiny$\rho\!=\!0.1$}
\put(64.5,56.5){\tiny$\rho\!=\!0.5$}
\put(64.5,51){\tiny$\rho\!=\!1$}
\put(64.5,46){\tiny$\rho\!=\!5$}
\end{overpic}
% }
 \hfill
% \fbox{
\begin{overpic}[width=0.66\textwidth,trim=0 40 30 40,clip]{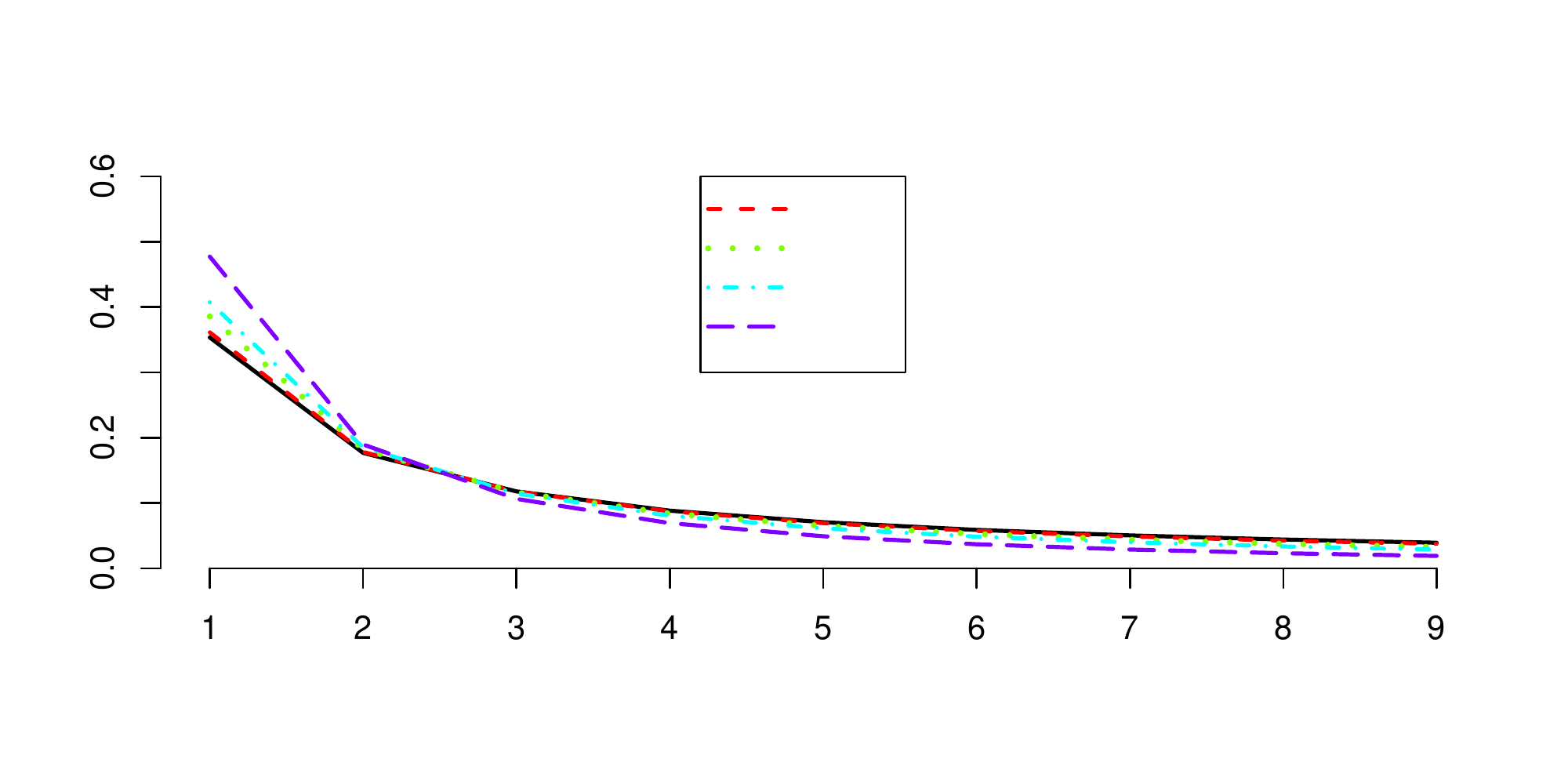}
\put(40,0){\footnotesize$s$ -- frequency of SNP}
%legend
\put(53.5,31){\tiny$\rho\!=\!0.1$}
\put(53.5,28.25){\tiny$\rho\!=\!0.5$}
\put(53.5,25.5){\tiny$\rho\!=\!1$}
\put(53.5,23){\tiny$\rho\!=\!5$}
\end{overpic}
% }
\caption{\label{sfsnorm} The expected site frequency spectrum from Figure \ref{sfsbunt} is shown, but has been normalized by the expected number of segregating sites. } 
\end{figure}
The higher proportion of low frequency site mutations in dispensable genes, results in negative values for Tajima's D, if the standard estimators $\widehat \theta_W$ and $\widehat \pi$ are used.
Using Tajima's D to detect purifying selection among dispensable gene sequences will thus be less accurate and lead to an increased number of false positives.
In Figure \ref{TajiDfig} a simulation of the effect on Tajima's D is shown.
\begin{figure}
 \includegraphics[width=0.5\textwidth]{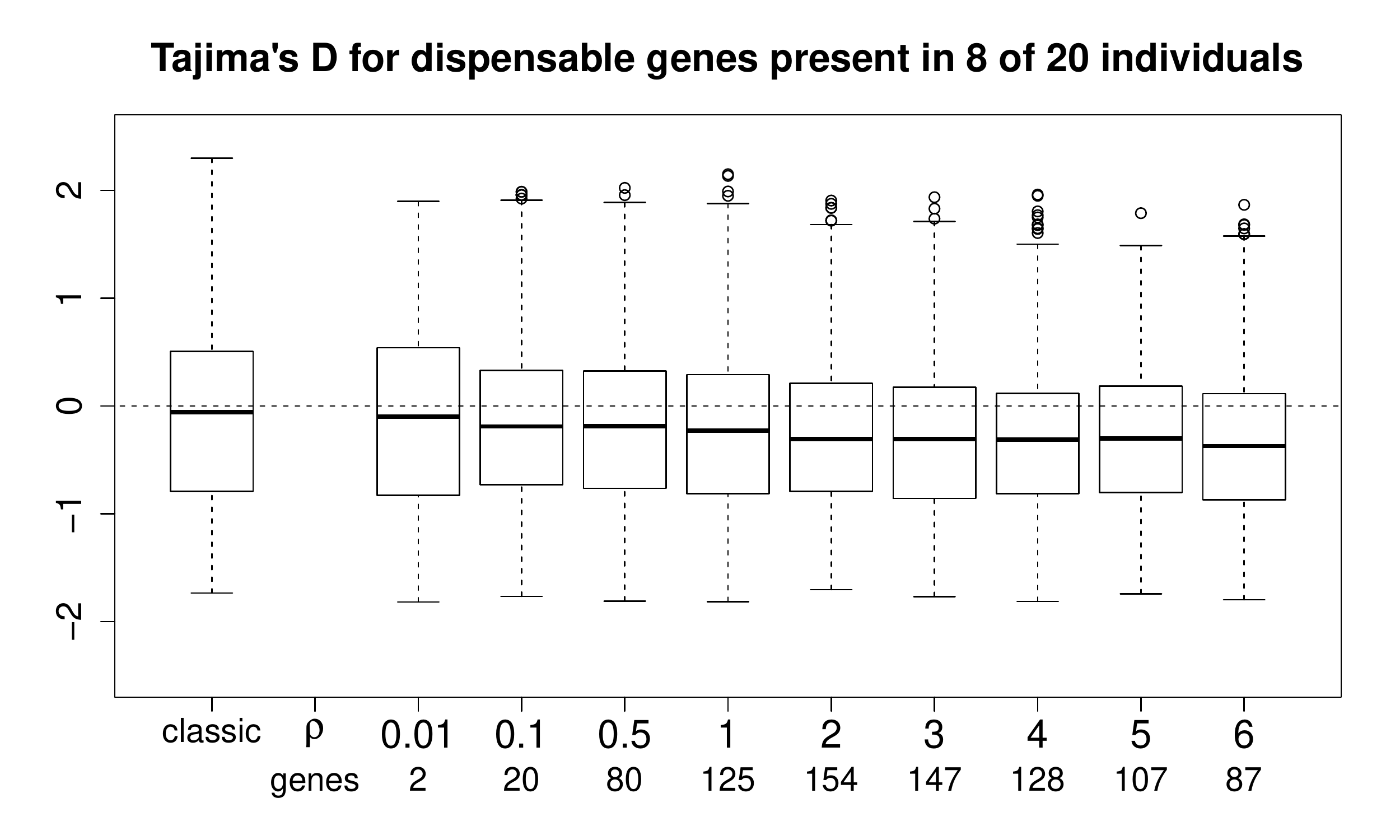} \includegraphics[width=0.5\textwidth]{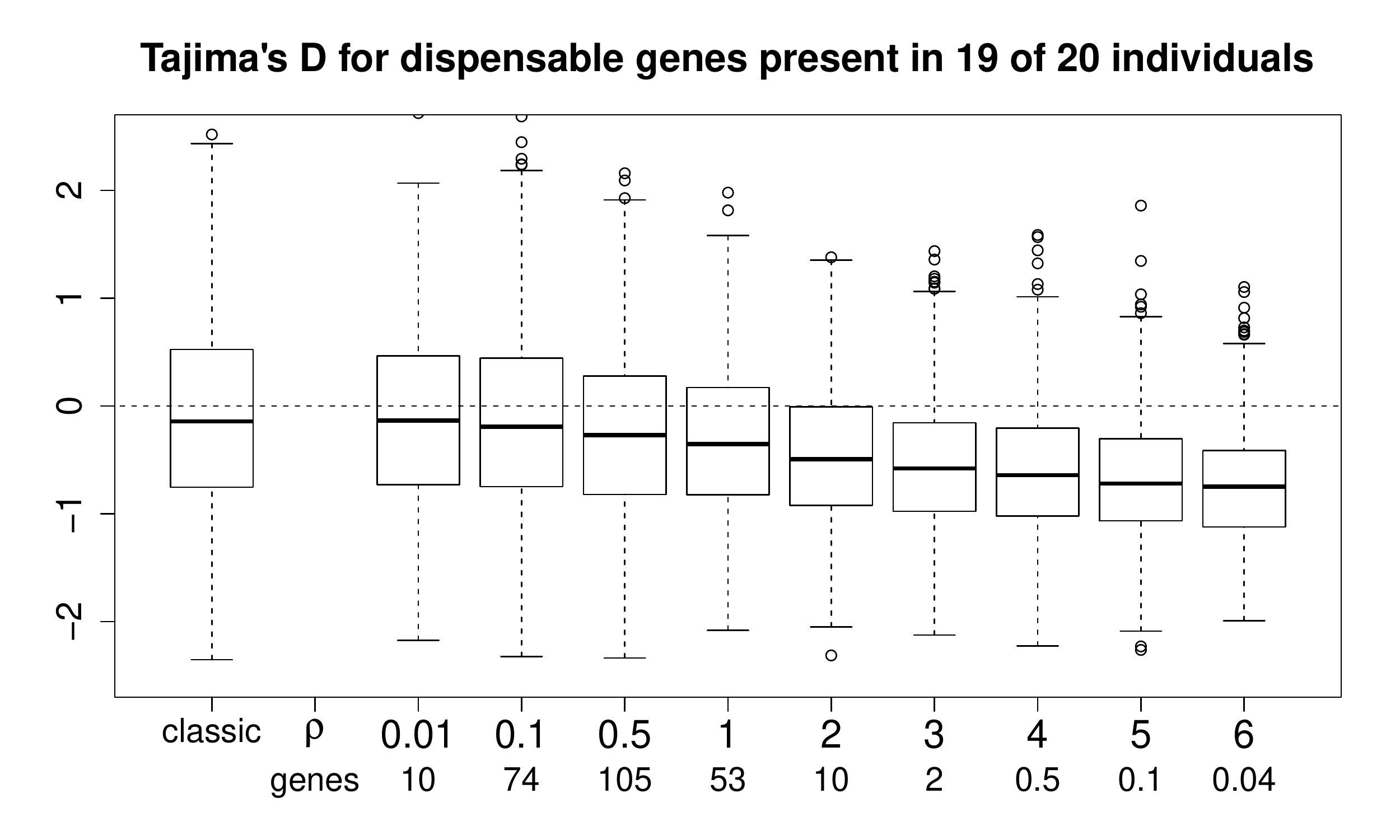}
\caption{Simulated values of Tajima's D in dispensable genes are shown for different values of $\rho$. The sample size is $n=20$ and the frequency of the dispensable gene is $k=8$ 
   on the left side, while $k=19$ on the right side. The site mutation rate $\theta_2 = 10000$ and the gene gain rate is given by $1000\rho$ such that the average number of genes in an individual is fixed at 1000.
   The average number of genes in frequency $k$ is shown in the second line beneath the according value of $\rho$.
   Note that results for arbitrary large (small) $\rho$ are somehow artificial as the number of dispensable genes in higher (lower) frequencies will converge to zero in this case. \label{TajiDfig}}
\end{figure}

% expected values are computed for independent trees so if one compares the Tajimas D of different gene sequences than this should still be valid 
% trees for different genes depend on each other
%different treees sliding window
Tajima's D is often used in a sliding window along a chromosome such that a general bias does not hide the outliers \citep{Oleksyk2010}.
While this is a reasonable procedure for essential core genes,
it is not clear whether one can identify outliers in dispensable gene sequences of different gene frequencies.
Even if we pool dispensable genes of frequency $k$ and compare Tajima's D between these genes, Tajima's D surely depends on the subset which possesses the gene.
Two dispensable genes of the same frequency $k \leq \tfrac{n}{2}$ might well be present in disjunct parts or in exactly the same subset of the sample.
Pooling all genes present in a certain subset of the population might thus be the more promising approach to compare Tajima's D between dispensable genes.

% Plotkin
Another frequently used statistic to detect genes under selection is the 
ratio of non-synonymous and synonymous SNPs \cite{Yang2000}.
In \cite{Kryazhimskiy2008} and \cite{Rocha2006} an effect occurring in closely related datasets is elaborated. %Kryazhimskiy and Plotkin showed that 
If individuals of a sample are too closely related, selection has not enough time to eliminate all negative mutations and thus the
power of dN/dS ratio tests, if at all applicable, is lower than in distinct lineage samples.
In this work we also considered closely related lineages, such that some genes are in intermediate frequency.
Thus we expect that similar issues will effect the dN/dS ratios in dispensable genes.
In contrast to the work of \citeauthor{Kryazhimskiy2008} and \citeauthor{Rocha2006} we focused on the general occurrence of neutral site mutations in dispensable genes and ignored selective effects.
Thus in our neutral framework the expected dN/dS ratio for a dispensable gene is always equal to one, regardless of the frequency of the gene.

\subsection{Conclusion}
Ignoring selective effects we allowed neutral 
genes sequences to get introduced and lost along the ancestral lineages, and inferred the resulting changes for the site frequency spectrum.
%IT HAS BEEN SHOWN PART
We showed that the site frequency spectrum for dispensable genes differs from the spectrum for essential core genes.
Our results reveal that frequently used estimators for the site mutation rate and tests, like Tajimas~D,
do not apply to dispensable gene sequences and will produce biased results.

The main incentive for this work was to improve analyses of site frequency spectra for prokaryotic dispensable genes.
% Übergang zu Eukaryoten über z.B. Arabidopsis thaliana
Nonetheless, the presented results should also apply to similar phenomena in eukaryotic DNA sequences.
For instance, a huge variation of genome sizes has been detected in the swedish \emph{Arabidopsis thaliana} population \citep{Long2013}. 
Our theory might well apply to the additional/dispensable parts within such genomes.
Although the model does not yet cover copies of gene sequences and each gained gene is a new one,
another relevant field might be the analyses of copy number variations (CNVs), e.g. in the human genome \citep{Redon2006,Freeman2006},
In fact similar results should hold for any kind of site frequency analyses,
where the considered individuals are correlated, in our case by gene gain, and form a, not necessarily nested, subtree within Kingman's coalescent.

\section{Proofs}
\label{proofs}

Some of the proofs in this Section will rely on Hoppe's urn model, which enables us to generate the $n$-coalescent forwards in time.
% The relationship between Hoppe's urn and Kingman coalescent, hit by mutational events (here gene loss), is by no means new.
This urn appeared first in \citep{Hoppe1984} and was used to show Ewen's sampling formula \citep{Ewens1972},
which gives the probability to sample a certain allele composition.
Here we will use Hoppe's urn to mimic the genealogy given by a Kingman coalescent hit by gene loss events, rather than allele mutations,
and compute the expected gene and site frequencies.

\subsection{Hoppe's urn}

Let us consider the history of one single gene $u$ along Kingman's coalescent, backwards in time, i.e. from the leaves to the root.
If we go backwards in time, we can not know at which time and at which lineage the gene $u$ has been gained, if at all.
Thus we consider all \emph{potential} gene loss events for the gene $u$ backwards in 
time, which happens at rate $\tfrac{\rho}{2}$ along the tree.
Only afterwards we will determine whether and if so, where the gene $u$ has been gained.
Now each of the potential loss events either ends up as an effectless event or alters the presence of gene $u$ at some ancestral lineage.
Assume a lineage is hit by a potential gene loss event. In this case if the gene is gained in the further past of this event, 
it will surely be absent in any of the descendants of this lineage.
So we only have to consider the gene gain events of gene $u$ along the unlost lineages
and we can discontinue to follow a lineage backwards in time when a potential gene loss event occurs.
At each of the potential gene loss events a random tree is rooted, such that
we obtain a forest of $K$ smaller trees, instead of one single tree.
The following definition describes the resulting forest.

\begin{definition}[Lineage loosing Kingman's coalescent]
\label{def:markedKingman}
  The lineage loosing Kingman coalescent $(R_t^{{\bf 0}})_{t \geq 0}$ is a continuous time Markov process with state space $\Pi_n^{{\bf 0}}$, 
the set of all marked partitions of $\{1,\dots,n\}$. The set of marked partitions is an extension of the set of partitions, where
 any partition element may or may not contain the additional element ${\bf 0}$.
A partition element $\xi_i$ of a partition $\xi = \{\xi_1,\dots,\xi_K\}$ is either a killed/lost element, if ${\bf 0} \in \xi_i$ or an active/unlost partition element
 if ${\bf 0} \notin \xi_i$. 
The infinitesimal generator $Q = (q_{\xi\eta})_{\xi,\eta \in \Pi_n}$ of the process is given by:\\
\begin{equation}
  q_{\xi\eta} = \begin{cases}
		  - \frac{k(k-1)+k\rho}{2} & \text{if } \xi = \eta\\
		  1 & \text{if } \xi \prec \eta\\ 
		  \frac{\rho}{2} & \text{if } \xi <_{{\bf 0}} \eta\\ 
		  0 & otherwise
                \end{cases}
\end{equation}
for $k := |\xi|_{{\bf \neg 0}}$, the number of partition elements $\xi_i$ in $\xi$, where ${\bf 0} \notin \xi_i$. We denote $\xi \prec \eta$ iff $\eta$ is obtained from 
$\xi$ by combining two partition elements without ${\bf 0}$ of $\xi$. We denote $ \xi <_{{\bf 0}} \eta $
iff $\eta$ is obtained by adding $\bf 0$ to one of the partition elements in $\xi$. 
The initial state $R_0 = \{ \{1\},\dots,\{n\} \}$ is the partition, where each $i \in \{1,\dots,n\}$ is its own partition element and does not contain $\bf 0$. 
\end{definition}

If we ignore the lengths the process $ (R_t^{{\bf 0}})_{t \geq 0} $ stays it is possible 
to construct the same forest forwards in time by the following P\'olya like urn model \citep{Hoppe1984}.

\begin{definition}[Hoppe's urn]
\label{def:hoppeurn}
Start with an urn with $i$ balls in $i$ different colors and one black ball.
Each colored ball has mass $1$ and the black ball has mass $\rho$.
Draw a new ball from the urn until there are $n$ colored balls within the urn.
\begin{itemize}
 \item If a colored ball is drawn, put the ball back into the urn together with an additional identical ball of the same color as the drawn ball.
 \item If the black ball is drawn, put the black ball back into the urn and add an additional ball of a new color to the urn.
\end{itemize}
Let $K \geq i$ be the random number of different colors among the final $n$ colored balls.
Denote by $\Pi_n^i = (n_1,\dots,n_K)$ the numbers of balls in the urn with the same color, when there are $n$ colored balls in total in the urn,
i.\,e.\ $n_j$ is the number of balls in color $j$ and there are $K$ different colors present in the urn.
We will assign the colors such that $n_1 \leq n_2 \cdots \leq n_K$.
\end{definition}

\begin{theorem}[Hoppe's urn describes the family size composition of Kingman's coalescent]
 \label{th:familysizehoppekingman}
 Let $\Pi_n^i$ be as in Definition \ref{def:hoppeurn} and 
 let $(R_j^{\bf 0})_{j=n,\dots,0}$ be the embedded Markov chain of $(R_t^{{\bf 0}})_{t \geq 0}$ from Definition \ref{def:markedKingman}.
 i.e. $R_j^{\bf 0}$ is the state of $(R_t^{{\bf 0}})_{t \geq 0}$, where $|R_t^{{\bf 0}}|_{{\bf \neg 0}} = j$.
 Further let $f$ be the function mapping the marked partition $R_j^{\bf 0} = \{\xi_1,\dots,\xi_K\}$ to the 
 family sizes $(n_1,\dots,n_K)$ such that  $n_l = |\xi_l \setminus \{{\bf 0}\} |$ and $n_1 \leq n_2 \leq \dots \leq n_K$.
 Then for any $a = (n_1,\dots,n_k)$
\begin{equation*}
 \mathbb P(\Pi_n^i = a) = \mathbb P(f(R_i^{\bf 0}) = a)
\end{equation*}
\end{theorem}

\begin{proof}
The embedded Markov chain $(R_j^{\bf 0})_{j=n,\dots,i}$ can be described by $n-i$ events $e_n,\dots,e_{i+1} \in \{\texttt{coal}, \texttt{loss} \}$.
The step from $j$ to $j-1$ will either merge two unlost partition elements of $R_j^{\bf 0}$ or loose an unlost partition element, by adding $\bf 0$ to it.
In any case the number of unlost lineages will be $|R_{j-1}^{{\bf 0}}|_{{\bf \neg 0}} = j-1$ at the next step.
Let the event from $R_j^{\bf 0}$ to $R_{j-1}^{\bf 0}$, denoted by $e_j$, be $e_j = \texttt{coal}$ if one of the $\binom{j}{2}$ pairs of unlost partition elements merges
and $e_j = \texttt{loss}$ if one of the $j$ unlost partition elements is marked by $\bf 0$.
We can then easily write down the probability 
\begin{equation} \frac{\prod\limits_{k=i+1}^n  \mathds{1}_{e_k = \texttt{coal}}  \frac{k(k-1)}{2} +     \mathds{1}_{e_k = \texttt{loss}} k \frac{\rho}{2}   } 
{\prod\limits_{k=i+1}^n \frac{k(k-1)}{2} + k \frac{\rho}{2}  }  
= 
 \frac{\rho^K \prod\limits_{k=i+1}^n  \mathds{1}_{e_k = \texttt{coal}} (k-1) } 
{\prod\limits_{k=i+1}^n (k-1 + \rho)  }  \label{pathprob}
\end{equation}
for a given sequence of events $e_n,\dots,e_{i+1}$, with $K$ loss events, to occur.

The process $(R_j^{\bf 0})_{j=n,\dots,0}$ runs backwards in time. We will now turn to the urn model, where
we can generate the same sequence of events $e_{i+1},\dots,e_n$ forwards in time.
If there are $j-1$ colored balls in the urn, the next chosen ball is colored with probability $\tfrac{j-1}{j-1+\rho}$, which equals
$$\frac{j-1}{j-1+\rho} = \frac{\binom{j}{2}}{\binom{j}{2} + j \frac{\rho}{2} }, $$
% $\tfrac{i}{i+\rho+\theta_1 du + \theta_2} = \frac{\binom{i+1}{2}}{\binom{i+1}{2} + (i+1)\frac{\rho}{2} + (i+1)\frac{\theta_1 du}{2} + (i+1)\frac{\theta_2}{2} } $.
the probability that two out of $j$ lineages coalesce before a lineage is killed off the tree by gene loss.
Let us denote the sequence of draws in Hoppe's urn by $\tilde e_{i+1},\dots,\tilde e_{n}$.
We write $\tilde e_j = \texttt{coal}$, if one of the $j$ colored balls is chosen and $\tilde e_j = \texttt{loss}$, if the black ball is chosen.
As $(\tilde e_j)_{i+1\leq j\leq n}$ has the same distribution as $(e_j)_{i+1\leq j\leq n}$,
the probability for a sequence $e_{i+1},\dots,e_n$ with $K$ loss events to occur is thus as well given by \eqref{pathprob}.

Finally note that each colored ball in the urn is equally likely chosen.
Just as each pair of lineages in the lineage loosing coalescent coalesces equally likely and each lineage gets lost with the same probability.
Thus the family sizes of the $i$ unlost lineages in $R_i^{\bf 0}$ 
%Kingman's lineage loosing coalescent
are distributed like the final numbers of balls in each of the $i$ colors,
Hoppe's urn started with.
\end{proof}

\begin{remark}
 While Hoppe's urn usually starts with one black ball, we allowed the urn to start with $i$ colored and one black ball.
\end{remark}

In order to capture the mutation dynamics within the urn model, we have to adapt Hoppe's urn to our model by adding marks and dots to the balls in the urn.

\begin{definition}
 Let $T_i$ be the length of the random time, where the lineage loosing coalescent has $i$ unlost lineages, i.e. 
 $T_i := \lambda(\{ t \in (0,\infty) : |R_t^{{\bf 0}}|_{{\bf \neg 0}} = i \})$.
\end{definition}

By construction, the random variables $(T_i)_{ i = 1,\dots,n}$ are independent exponentially distributed waiting times with mean $\mathbb E[T_i] = \frac{2}{i(i-1)+i\rho}$.

In the next definition we will distinguish between marked and unmarked balls.
In addition to a mark each ball can carry any number of colored dots. 
Please note, a mark does here not mean that a lineage is lost, as all balls in the urn correspond to unlost lineages.
Instead a mark will correspond to a gene gain event, while dots correspond to site mutations.

\begin{remark}
Note that, in the lineage loosing coalescent, only the history of one possible gene $u$ is described.
However, new genes are picked from $I$ according to the Lebesgue measure $\lambda$.
Such that a gene $u \in I$ is almost surely never gained.
To compute first moment statistics within Hoppe's urn, we will use an abused notation.
We will add a gene gain of a gene in an arbitrary small interval $du  \subset I, |du| \searrow 0$ to the description of Hoppe's urn such that a gene gain event in $du$ occurs at rate 
$\frac{\theta_1}{2}du$. Since at any time there are only finitely many genes present in the ancestral lineages of the sample and $du$ is small,
there will be at most one gene gain in $du$ along the lineages.
On the over hand, we will consider all possible site mutations $(u,v) \in \{u\} \times J$ at once,
such that the site mutation rate is $\frac{\theta_2}{2}$, once the gene $u$ in $du$ has been gained.
\end{remark}
% In order to regain higher moment statistics we will have to consider all possible genes in $[0,1]$ at once, including their mutations.  --> Verweis auf HGT paper

\begin{definition}[Hoppe's urn with colored dots]
 \label{def:urnwithdots}
If there are $i$ colored balls in Hoppe's urn wait an exponential time $T_i$ with mean $\frac{i(i-1)}{2}+i\frac{\rho}{2}$ until drawing the next ball.
During this waiting time mark each of the colored balls at rate $\frac{\theta_1}{2}du$. 
As soon as a ball is marked place a dot in a new color $v$ at rate $\frac{\theta_2}{2}$ to this ball.
If after $T_i$:
\begin{itemize}
 \item a colored ball is drawn, put the ball back into the urn and add an additional ball with identical color, mark and dots as the drawn ball.
 \item a black ball is drawn, put the ball back into the urn and add a ball of a new color, without mark or any dots, to the urn.
\end{itemize}
Finally, if there are $n$ colored balls in the urn wait an exponential waiting time $T_n$ with mean $\frac{n(n-1)}{2}+n\frac{\rho}{2}$, place the corresponding marks and dots then stop.
\end{definition}

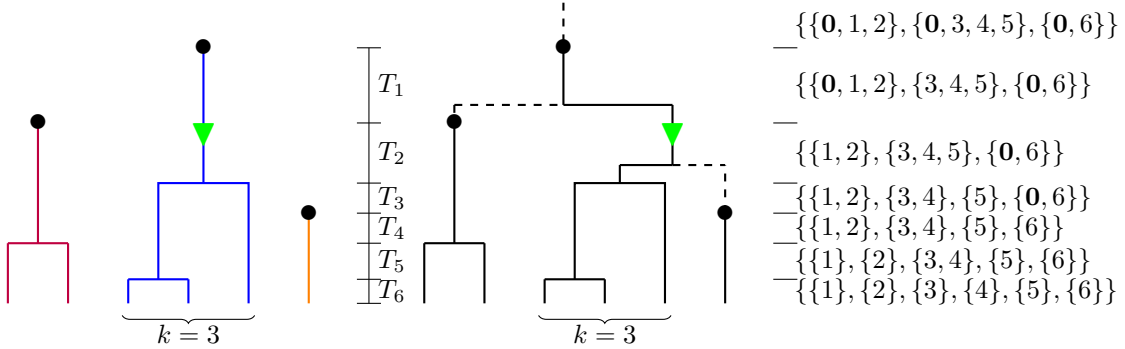
\begin{figure}
\begin{tikzpicture}[scale = 0.8]
% coalescent
%vertical line above MRCA
% \draw (2.875,4) -- (2.875,5);
% left group (2)
 \draw[color = purple,thick] (2,1)--(2,0);
\draw[color = purple,thick] (1,1)--(1,0);
 \draw[color = purple,thick] (2,1)--(1,1);
%right group (3)
\draw[color = blue,thick] (3,0.4)--(3,0);
\draw[color = blue,thick] (4,0.4)--(4,0);
 \draw[color = blue,thick] (3,0.4)--(4,0.4);
\draw[color = blue,thick] (3.5,0.4) -- (3.5,2) -- (5,2) -- (5,0);

%join groups
% \draw (1.5,1) -- (1.5,4) -- (4.25,4) -- (4.25,2);

%additional line
\draw[color = orange,thick] (6,1.5) -- (6,0);
% \draw[fill, color = orange] (6,1.5) circle (0.075);
\draw[font = \Large] (6,1.5) node (hoppeloss1) {$\bullet$};

\draw[color = purple,thick] (1.5,1) -- (1.5,3);
% \draw[fill, color = purple] (1.5,3) circle (0.075);
\draw[font = \Large] (1.5,3) node (hoppeloss2) {$\bullet$};
\draw[color = blue,thick] (4.25,2) -- (4.25,4.25);
% \draw[fill, color = blue] (4.25,4.25) circle (0.075);
\draw[font = \Large] (4.25,4.25) node (hoppeloss3) {$\bullet$};

\draw[font = \Large, color = green] (4.25,2.8) node (hoppemark1) {$\blacktriangledown$};

\draw[decorate,decoration={brace,amplitude=3pt,mirror}] 
    (2.9,-0.2) -- (5.1,-0.2); 

\draw node at (4,-0.5){$k = 3$};

% LEiste mit Zeiten
\draw (7,0) -- (7,4.25);
\draw (6.8,0) -- (7.2,0);
\draw[right] node at (7,0.2){$T_6$};
\draw (6.8,0.4) -- (7.2,0.4);
\draw[right] node at (7,0.7){$T_5$};
\draw (6.8,1) -- (7.2,1);
\draw[right] node at (7,1.25){$T_4$};
\draw (6.8,1.5) -- (7.2,1.5);
\draw[right] node at (7,1.75){$T_3$};
\draw (6.8,2) -- (7.2,2);
\draw[right] node at (7,2.5){$T_2$};
\draw (6.8,3) -- (7.2,3);
\draw[right] node at (7,3.625){$T_1$};
\draw (6.8,4.25) -- (7.2,4.25);
 \end{tikzpicture}
\begin{tikzpicture}[scale = 0.8]
% coalescent
%vertical line above MRCA
% \draw (2.875,4) -- (2.875,5);
% left group (2)
 \draw[thick] (2,1)--(2,0);
\draw[thick] (1,1)--(1,0);
 \draw[thick] (2,1)--(1,1);
%right group (3)
\draw[thick] (3,0.4)--(3,0);
\draw[thick] (4,0.4)--(4,0);
 \draw[thick] (3,0.4)--(4,0.4);
\draw[thick] (3.5,0.4) -- (3.5,2) -- (5,2) -- (5,0);

%join groups
% \draw (1.5,1) -- (1.5,4) -- (4.25,4) -- (4.25,2);

%additional line
\draw[thick] (6,1.5) -- (6,0);
\draw[thick,dashed] (6,1.5) -- (6,2.3);
\draw[thick] (5.125,2.3) -- (4.25,2.3);
\draw[thick,dashed] (5.125,2.3) -- (6,2.3);
\draw[thick] (5.125,2.3) -- (5.125,3.3);
\draw[thick] (3.3125,3.3) -- (3.3125,4.25);
\draw[thick,dashed] (3.3125,5) -- (3.3125,4.25);

% \draw[fill, color = orange] (6,1.5) circle (0.075);
\draw[font = \Large] (6,1.5) node (hoppeloss1) {$\bullet$};

\draw[thick] (1.5,1) -- (1.5,3);
\draw[thick,dashed] (1.5,3.3) -- (1.5,3);
\draw[thick,dashed] (1.5,3.3) -- (3.3125,3.3);
\draw[thick] (3.3125,3.3) -- (5.125,3.3);
% \draw[fill, color = purple] (1.5,3) circle (0.075);
\draw[font = \Large] (1.5,3) node (hoppeloss2) {$\bullet$};
\draw[thick] (4.25,2) -- (4.25,2.3);
% \draw[fill, color = blue] (4.25,4.25) circle (0.075);
\draw[font = \Large] (3.3125,4.25) node (hoppeloss3) {$\bullet$};

\draw[font = \Large, color = green] (5.125,2.8) node (hoppemark1) {$\blacktriangledown$};

\draw[decorate,decoration={brace,amplitude=3pt,mirror}] 
    (2.9,-0.2) -- (5.1,-0.2); 

\draw node at (4,-0.5){$k = 3$};

% LEiste mit coalescent
% \draw (7,0) -- (7,4.25);
% \draw (6.8,0) -- (7.2,0);
\draw[right] node at (7,0.2){$\{\{1\},\{2\},\{3\},\{4\},\{5\},\{6\}\}$};
\draw (6.8,0.4) -- (7.2,0.4);
\draw[right] node at (7,0.7){$\{\{1\},\{2\},\{3,4\},\{5\},\{6\}\}$};
\draw (6.8,1) -- (7.2,1);
\draw[right] node at (7,1.25){$\{\{1,2\},\{3,4\},\{5\},\{6\}\}$};
\draw (6.8,1.5) -- (7.2,1.5);
\draw[right] node at (7,1.75){$\{\{1,2\},\{3,4\},\{5\},\{{\bf 0},6\}\}$};
\draw (6.8,2) -- (7.2,2);
\draw[right] node at (7,2.5){$\{\{1,2\},\{3,4,5\},\{{\bf 0},6\}\}$};
\draw (6.8,3) -- (7.2,3);
\draw[right] node at (7,3.625){$\{\{{\bf 0},1,2\},\{3,4,5\},\{{\bf 0},6\}\}$};
\draw[right] node at (7,4.625){$\{\{{\bf 0},1,2\},\{{\bf 0},3,4,5\},\{{\bf 0},6\}\}$};
\draw (6.8,4.25) -- (7.2,4.25);
 \end{tikzpicture}
\caption{A realization of Hoppe's urn is shown. During $T_2$ the first line is of size $k=2$ and the second line is of size $3$. 
For $i=k=2$ and $m=1$, $\mathcal T(2,2,1)$ as given above \eqref{ElT} equals
 $\{2,3,4\}$, while $\mathcal T(2,2,2) = \{5,6\}$. The corresponding lineage loosing Kingman's coalescent is shown on the left side. 
 Lost lineages are shown as dashed lines, in addition we show the would-be merging points if the lineages would not have been lost.  \label{hoppebild} }
\end{figure}

Denote by $B_i$ the $i$-th colored ball added to the urn.
We may just as in Definition \ref{def:treeindexedMC} represent the final mark and dots on ball $B_i$ by 
a simple finite counting measure $D_{du}(B_i) \in \mathcal N_f(\{u\} \times ( \{0\} \cup J))$, for $J= (0,1]$.
If $B_i$ is a marked set $D_{du}(B_i)(u,0) = 1$, otherwise $D_{du}(B_i)(u,0) = 0$.
In addition if a dot in a new color $v$ is placed to $B_i$, choose $v$ according to the Lebesgue measure from $J = (0,1]$.

\begin{theorem}
\label{th:whyhoppeworks}
Consider $\mathcal M_i \in \mathcal N_f(I \times J)$ as given in Definition \ref{def:treeindexedMC} for $i=1,\dots,n$.
Let $f: \mathcal N_f(I \times J)^n \rightarrow \mathbb R$ be a summary statistic of the sample $(\mathcal M_i)_{i = 1,\dots,n}$. 
Assume further that $f$ does only involve properties of single genes. That is, for all $M \in \mathbb N$:
 \begin{equation*}
  f((\mathcal M_i)_{i = 1,\dots,n}) = \sum_{m=0}^{M-1} f( (\mathcal M_i|_{\{(\frac{m}{M},\frac{m+1}{M}]\}\times J})_{i = 1,\dots,n}  ).
 \end{equation*}
 or in a more convenient, though abused notation
\begin{equation*}
  f((\mathcal M_i)_{i = 1,\dots,n}) = \int_I f( (\mathcal M_i|_{du\times J})_{i = 1,\dots,n}  )
 \end{equation*}
holds.
Moreover, let $D(B_i)$ be the finite counting measures described above, given by Hoppe's urn started with one black ball, then
\begin{equation}
 \mathbb E[ f((\mathcal M_i)_{i = 1,\dots,n}) ] = \int_I \mathbb E[ f( (D_{du}(B_i))_{i=1,\dots,n} ) ].
\end{equation}
\end{theorem}

\begin{proof}
In the tree indexed Markov chain from Definition \ref{def:treeindexedMC} site mutations within dispensable genes can only occur if the gene already exists.
Thus, before we start adding site mutations in $du \times J$ we have to wait until a gene $u \in du$ is gained.
This corresponds to placing a mark at rate $\tfrac{\theta_1}{2} du$ to a ball during the evolution of the urn.
As we are waiting an exponential time $T_i$ with mean $\tfrac{2}{i(i-1)+i\rho}$ until drawing the next ball this gives us the correct branch lengths within the forest
and since $du$ is small there is at most one mark until the urn is stopped.
Furthermore Theorem \ref{th:familysizehoppekingman} ensures that the family sizes of each ball/lineage are equally distributed.

Thus placing colored dots at rate $\tfrac{\theta_2}{2}$ at each marked ball gives the same distribution, as the site mutations within the gene $u$, which 
occur at rate $\tfrac{\theta_2}{2}$ for each lineage, which carries the gained gene $u$ in $du$. So we get 
\begin{equation*}
 \mathcal M_i|_{du \times J} \stackrel{d}{=} D_{du}(B_i),
\end{equation*}
and finally
\begin{align*}
 \mathbb E[ f((\mathcal M_i)_{i = 1,\dots,n}) ] =  \int_I \mathbb E[ f( (\mathcal M_i|_{du\times J})_{i = 1,\dots,n} ) ] = \int_I \mathbb E[ f( (D_{du}(B_i))_{i=1,\dots,n} ) ]. \qedhere
\end{align*}
\end{proof}

\begin{remark}
% Note that, Theorem \ref{th:whyhoppeworks} does only hold for the first moment of statistics involving properties of the genes itself, and no pairs of genes.
In particular, Theorem \ref{th:whyhoppeworks} holds for the joint site and gene frequency spectrum 
$$f((\mathcal M_i)_{i = 1,\dots,n}) := G_{k,s}$$
 from Definition \ref{def:jointfs}.
We will use this relationship between Hoppe's urn and the tree indexed Markov chain to investigate $\mathbb E[G_{k,s}]$ in the proof of Theorem \ref{jointfs}.
Note that, Theorem \ref{th:whyhoppeworks} does not hold for higher moments and statistics involving pairs of genes, as Hoppe's urn is unable to mimic the 
joint genealogy of two genes.
\end{remark}

\subsection{Proof of Theorem \ref{jointfs}}

\begin{proof}[Proof of Theorem \ref{jointfs}]

To show Theorem \ref{jointfs} we will take the same route as in the proof of Theorem~5 in \cite{BaumdickerHessPfaffelhuber2010}.\\
Based on Hoppe's urn with colored dots we will compute

\begin{align}
\begin{split}
\label{ausgang}
\mathbb E[G_{k,s}] = \int_I \int_J \mathbb E\big[ (d&u,0) \in \mathcal M_i \text{ for exactly $k$ different $i$}  \\
                       &\text{and } (du,dv) \in \mathcal M_i \text{ for exactly $s$ different $i$}\big]
\end{split}
\end{align}

%TODO explain parts of the formula

In Hoppe's urn we say that line $l$ during $T_i$ is of size $k$ if the ball belonging to this line
produces exactly $k-1$ offspring until the urn finishes. Then we can express \eqref{ausgang} as

\begin{align*}
\mathbb E[G_{k,s}] =
 & \sum_{i=1}^n \sum_{l=1}^i \mathbb P[\text{$l$th line during $T_i$ is of size $k$}] \cdot \int_I \mathbb P[\text{mark in $du$ on $l$th line during $T_i$}]\\
 & \quad  \cdot \int_J \mathbb E\left[ (du \times dv) \in \mathcal M_{i_j} \text{ for exactly } s \text{ different } i_j \Big | 
	    \begin{array}{l} 
	    \text{$l$th line during $T_i$ is of size $k$}\\ \text{and mark in $du$ on $l$th line during $T_i$}\\
	    \end{array}
 \right].
\end{align*}

The first two terms in the sum are already known from the proof of Theorem~\ref{genefreqspec}. In \cite{BaumdickerHessPfaffelhuber2010} we showed that
\begin{equation*}
 \mathbb P[\text{mark in $du$ on $l$th line during $T_i$}] = \frac{\theta_1}{i(i-1+\rho)}du
\end{equation*}
as $T_i$ is exponential distributed with mean $\tfrac{2}{j(j-1+\rho)}$ and gene gains in $du$ occur at rate $\theta_1 du$.
Using Hoppe's urn it is possible to show that
\begin{equation}
\label{oldhoppe}
\mathbb P[\text{$l$th line during $T_i$ is of size $k$}] = \binom{n-i}{k-1} \frac{(k-1)!(i-1+\rho) \cdots (n-k-1+\rho)}{(i+\rho) \cdots (n-1+\rho) }.
\end{equation}

% which led to
% \begin{align*}
%   \mathbb E[G_k] & = \sum_{i=1}^n \sum_{l=1}^i
%     \mathbb P[\text{$l$th line during $T_i$ is of size $k$}] \cdot
%     \int \mathbb E[\text{mark in $du$ on $l$th line during $T_i$}]\\
%     & = \frac{\theta_1}{k}\frac{n\cdots(n-k+1)}{(n-k+\rho)\cdots(n-1+\rho)}
% \end{align*}

Note that we will use the notation $(a)\cdots(a+b-1)$ for the Pochhammer function $(a)_b = (a)(a+1)(a+2)\cdots(a+b-1)$, and if $b<1$, then $(a)\cdots(a+b-1) = 1$.

So we are left with the last factor 
\begin{align}
 \mathbb E\left[ (du \times dv) \in \mathcal M_{i_j} \text{ for exactly } s \text{ different } i_j \Big | 
	    \begin{array}{l}
	    \text{$l$-th line during $T_i$ is of size $k$}\\ \text{and mark in $du$ on $l$-th line during $T_i$}\\
	    \end{array}
 \right]
\end{align}
which again remains a combinatorical problem in Hoppe's urn, if we add the possibility to place dots at the balls of the urn as given in Definition \ref{def:urnwithdots}.
A ball within this urn which is marked corresponds to a gene gain event, while placing dots at marked balls corresponds to single point mutations.

We denote by $T_i$ the random waiting time where there are $i$ colored balls present in Hoppe's urn. 
This corresponds in Kingman's coalescent to the time $i$ lineages need to coalesce into $i-1$ lineages or 
loose one of the $i$ lines at rate $\tfrac{\rho}{2}$.
Suppose during $T_i$ one of the balls is marked,
then the time after this event is again distributed like $T_i$ due to the memoryloss of the exponential distribution.
Now let us consider only realizations of Hoppe's urn, 
where there are $k$ marked balls among the final $n$ balls.
Let $\mathcal T(i,k,m)$ be the set of all  $j \in \{i,\dots,n\}$ where $m$ of $j$ colored balls in the urn are marked and
let $l(\mathcal T) := \sum_{j \in \mathcal T} T_j$. See also Figure \ref{hoppebild} for an illustration.
Then
\begin{align}\label{ElT}
 \mathbb E [ l(\mathcal T(i,k,m)) ] = \sum_{j=i}^n \mathbb P[ T_j \in \mathcal T(i,k,m) ] \mathbb E[ T_j ].
\end{align}

Starting with $i-1$ unmarked balls and one marked ball there are 
 $\binom{n-i}{k-1}$ possibilities at what times $k-1$
  marked balls are added when $n-i$ balls are added to the urn in
  total. And the probability for any of these possibilities equals in Hoppe's urn.
 To get \eqref{ElT} one has to consider how many of these possibilities will have $m$ marked balls at time $j$.
 Clearly this is only the case if $m-1$ marked balls are added when adding the first $j-i$ balls and the remaining $k-m$
 marked balls are added within the remaining $n-j$ balls.
Thus 
\begin{align*}
 \mathbb P(j \in \mathcal T(i,k,m) ) = \binom{j-i}{m-1} \cdot \binom{n-j}{k-m} \cdot \binom{n-i}{k-1}^{-1}
\end{align*}
and
\begin{align*}
 \mathbb E [ l(\mathcal T(i,k,m)) ] = \sum_{j=i}^n \binom{j-i}{m-1} \cdot \binom{n-j}{k-m} \cdot \binom{n-i}{k-1}^{-1}  ( \tfrac j2(j-1+\rho) )^{-1}.
\end{align*}

In the subtree only mutations/dots gained at a line of size $s$ result in SNPs in frequency $s$, see Figure \ref{subtree}.
Note that the subtree does not contain any loss events.
If there are $m$ marked balls, we still have to add $k-m$ marked balls and if $s-1$ of these balls have to be the offspring of one ball, then 
there are $\binom{k-m}{s-1}$ possibilities to do so. Each of these possibilites has probability $\frac{(s-1)!(m-1)\cdots(k-s-1)}{(m)\cdots(k-1)}$.
Thus, for each of the $m$ marked balls during any time with $j \in \mathcal T(i,k,m)$ colored balls the probability to have $s-1$ offspring is given by 
\begin{align}
\label{newhoppe}
  \binom{k-m}{s-1} \frac{(s-1)!(m-1)\cdots(k-s-1)}{(m)\cdots(k-1)}.
\end{align}
In contrast to equation \eqref{oldhoppe} the above formula does not contain any $\rho$, as we are currently acting only in the marked subtree of the forest produced by Hoppe's urn.

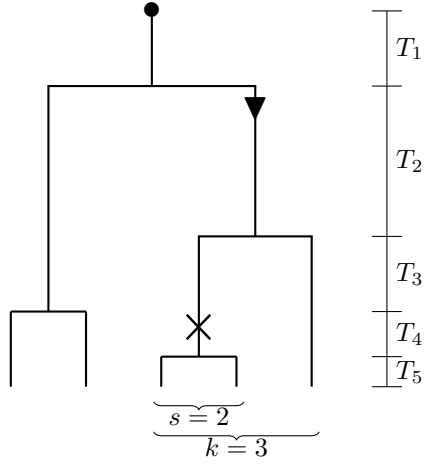
\begin{figure}
\begin{center}
\begin{tikzpicture}[scale = 1]
% coalescent
%loss at root 
\draw[font = \Large] (2.875,5) node (theloss) {$\bullet$};
%vertical line above MRCA
\draw[thick] (2.875,4) -- (2.875,5);
% left group (2)
 \draw[thick] (2,1)--(2,0);
\draw[thick] (1,1)--(1,0);
 \draw[thick] (2,1)--(1,1);
%right group (3)
\draw[thick] (3,0.4)--(3,0);
\draw[thick] (4,0.4)--(4,0);
 \draw[thick] (3,0.4)--(4,0.4);
\draw[thick] (3.5,0.4) -- (3.5,2) -- (5,2) -- (5,0);
%join groups
\draw[thick] (1.5,1) -- (1.5,4) -- (4.25,4) -- (4.25,2);

\draw[decorate,decoration={brace,amplitude=3pt,mirror}] 
    (2.9,-0.6) -- (5.1,-0.6); 
\draw node at (4,-0.8){$k = 3$};

\draw[decorate,decoration={brace,amplitude=3pt,mirror}] 
    (2.9,-0.2) -- (4.1,-0.2); 
\draw node at (3.5,-0.4){$s = 2$};

% LEiste mit Zeiten
\draw (6,0) -- (6,5);
\draw (5.8,0) -- (6.2,0);
\draw[right] node at (6,0.2){$T_5$};
\draw (5.8,0.4) -- (6.2,0.4);
\draw[right] node at (6,0.7){$T_4$};
\draw (5.8,1) -- (6.2,1);
\draw[right] node at (6,1.5){$T_3$};
\draw (5.8,2) -- (6.2,2);
\draw[right] node at (6,3){$T_2$};
\draw (5.8,4) -- (6.2,4);
\draw[right] node at (6,4.5){$T_1$};
\draw (5.8,5) -- (6.2,5);

\draw[font = \Large] (4.25,3.7) node (gain1) {$\blacktriangledown$};

\draw (3.5,0.8) node (mut11) {\XSolid};
 \end{tikzpicture}
 \caption{
  The second line during $T_2$ is of size $k=3$ and there is a mark $(\blacktriangledown)$ on the second line during $T_2$. Due to the dot $(\times)$ during $T_4$ at the first line of the marked subtree,
  the corresponding mutation is present in $s=2$ out of $k=3$ individuals which carry the gene. \label{subtree}}
\end{center}
\end{figure}

Combining the above yields
\begin{align*}
 &\mathbb E\left[ (du \times dv) \in \mathcal M_{i_j} \text{ for exactly } s \text{ different } i_j \Big | 
	    \begin{array}{l}
	    \text{$l$th line during $T_i$ is of size $k$}\\ \text{and mark in $du$ on $l$th line during $T_i$}\\
	    \end{array}
 \right]\\
=&\sum_{m=1}^k \sum_{r=1}^m \mathbb P [\text{$r$th line during $\mathcal T(i,k,m)$ is of size $s$}] \cdot \int \mathbb P[\text{ dot in $dv$ on $r$th line during $\mathcal T(i,k,m)$}]\\
=&\sum_{m=1}^k m \binom{k-m}{s-1} \frac{(s-1)!(m-1)\cdots(k-s-1)}{(m)\cdots(k-1)}      \sum_{j=i}^n \binom{j-i}{m-1}  \binom{n-j}{k-m} \binom{n-i}{k-1}^{-1}          \frac{\theta_2}{j(j-1+\rho)}dv\\
=&\sum_{m=1}^{k-s+1} m \binom{k-m}{s-1} \frac{(s-1)!(m-1)\cdots(k-s-1)}{(m)\cdots(k-1)}      \sum_{j=i}^n \binom{j-i}{m-1}  \binom{n-j}{k-m} \binom{n-i}{k-1}^{-1}          \frac{\theta_2}{j(j-1+\rho)}dv\\
% &= \sum_{m=1}^{k-s} m \binom{k-m}{s-1} \frac{(s-1)!(m-1)\cdots(k-s-1)}{(m)\cdots(k-1)}      \sum_{j=i}^n \binom{j-i}{m-1}  \binom{n-j}{k-m} \binom{n-i}{k-1}^{-1}          \frac{\theta_2}{j(j-1+\rho)}dv\\
% &\quad +    (k-s+1) \frac{(s-1)!}{(k-s+1)\cdots(k-1)}   \sum_{j=i}^n \binom{j-i}{k-s}  \binom{n-j}{s-1} \binom{n-i}{k-1}^{-1}          \frac{\theta_2}{j(j-1+\rho)}dv\\
% &= \frac{\theta_2}{s} \Big( \sum_{m=1}^{k-s} \binom{k-m}{s-1} \binom{k-1}{s}^{-1}      \sum_{j=i}^n \binom{j-i}{m-1}  \binom{n-j}{k-m} \binom{n-i}{k-1}^{-1}          \frac{1}{j(j-1+\rho)}dv\\
% &\quad +    (k-s+1) \frac{(s)!}{(k-s+1)\cdots(k-1)}   \sum_{j=i}^n \binom{j-i}{k-s}  \binom{n-j}{s-1} \binom{n-i}{k-1}^{-1}          \frac{1}{j(j-1+\rho)}dv \Big)\\
\end{align*}

Setting all parts together gives
\begin{align*}
 \mathbb E[G_{k,s}] =
 & \sum_{i=1}^n \sum_{l=1}^i  \binom{n-i}{k-1}
     \frac{(k-1)!(i-1+\rho)\cdots(n-k-1+\rho)}{(i+\rho)\cdots(n-1+\rho)}
     \frac{\theta_1}{i(i-1+\rho)}\\
 & \quad \frac{\theta_2}{s}  \sum_{m=1}^{k-s+1} m \frac{s!(m-1)\cdots(k-s-1)}{(m)\cdots(k-1)} \binom{k-m}{s-1}       
     \sum_{j=i}^n \binom{j-i}{m-1} \binom{n-j}{k-m} \binom{n-i}{k-1}^{-1}          \frac{1}{j(j-1+\rho)}\\
 = & \frac{\theta_1}{k} \frac{k!}{(n-k+\rho)\cdots(n-1+\rho)} \frac{\theta_2}{s}  \\
    &
	\sum_{m=1}^{k-s+1} m \frac{s!(m-1)\cdots(k-s-1)}{(m)\cdots(k-1)} \binom{k-m}{s-1}  \sum_{i=1}^n \sum_{j=i}^n \binom{j-i}{m-1} \binom{n-j}{k-m} \frac{1}{j(j-1+\rho)}
\end{align*}

\allowdisplaybreaks
% We can simplify the sum to:
% \begin{align*}
% &\sum_{m=1}^k \frac{s!(m)\cdots(k-s-1)}{(m)\cdots(k-1)} \binom{k-m}{s-1}  \sum_{j=1}^n \binom{n-j}{k-m} \frac{1}{j(j-1+\rho)} \sum_{i=1}^j \binom{j-i}{m-1} \\
% =& \sum_{m=1}^k \binom{k-1}{s}^{-1} \binom{k-m}{s-1}  \sum_{j=1}^n \binom{n-j}{k-m} \frac{1}{j(j-1+\rho)} \binom{j}{m} \\
% =& \sum_{j=1}^n \frac{1}{j(j-1+\rho)} \sum_{m=1}^k \binom{k-1}{s}^{-1} \binom{k-m}{s-1}   \binom{n-j}{k-m}  \binom{j}{m} \\
% \end{align*}
% NO, 
We can simplify the sum if $s < k$, as then $  m \frac{s!(m-1)\cdots(k-s-1)}{(m)\cdots(k-1)} \binom{k}{s-1} = \frac{m(m-1)sk}{(k-s)(k-s+1)} $ to
\begin{align*}
&\sum_{m=1}^{k-s+1} m \frac{s!(m-1)\cdots(k-s-1)}{(m)\cdots(k-1)} \binom{k-m}{s-1}  \sum_{j=1}^n \binom{n-j}{k-m} \frac{1}{j(j-1+\rho)} \sum_{i=1}^j \binom{j-i}{m-1} \\
&= \sum_{m=1}^{k-s+1} m \frac{s!(m-1)\cdots(k-s-1)}{(m)\cdots(k-1)} \binom{k-m}{s-1}  \sum_{j=1}^n \binom{n-j}{k-m} \frac{1}{j(j-1+\rho)} \binom{j}{m} \\
&= \sum_{j=1}^n \frac{1}{j(j-1+\rho)} \sum_{m=1}^{k-s+1} m \frac{s!(m-1)\cdots(k-s-1)}{(m)\cdots(k-1)} \binom{k-m}{s-1}   \binom{n-j}{k-m}  \binom{j}{m} \\
&= \sum_{j=1}^n \frac{1}{j(j-1+\rho)} \sum_{m=1}^{k-s+1} m \frac{s!(m-1)\cdots(k-s-1)}{(m)\cdots(k-1)} \binom{k}{s-1} \binom{n}{j}^{-1} \binom{k-s+1}{m}   \binom{n-k}{j-m}  \binom{n}{k}  \\
&= \binom{n}{k}  \frac{sk}{(k-s)(k-s+1)} \sum_{j=1}^{n-s+1} \binom{n}{j}^{-1} \frac{1}{j(j-1+\rho)} \sum_{m=1}^{k-s+1}   m (m-1)   \binom{k-s+1}{m}   \binom{n-k}{j-m}   \\
&= \binom{n}{k}  sk \sum_{j=1}^{n-s+1} \binom{n}{j}^{-1} \frac{1}{j(j-1+\rho)} \sum_{m=2}^{k-s+1}    \binom{k-s-1}{m-2}   \binom{n-k}{j-m}   \\
&= \binom{n}{k}  sk \sum_{j=1}^{n-s+1} \binom{n}{j}^{-1} \frac{1}{j(j-1+\rho)} \binom{n-s-1}{j-2}   \\
&= \binom{n}{k}  \frac{sk}{(n-s)(n-s+1)} \sum_{j=0}^{n-s-1} \binom{n}{j+2}^{-1} \frac{j+1}{j+1+\rho} \binom{n-s+1}{j+2}   \\
&= \binom{n}{k}  \frac{k}{n} \binom{n-1}{s}^{-1} \sum_{j=0}^{n-s-1} \frac{j+1}{j+1+\rho} \binom{n-j-2}{s-1}   \\
\end{align*}

If $k=s$ we have $  m \frac{s!(m-1)\cdots(k-s-1)}{(m)\cdots(k-1)} \binom{k}{s-1} = sk$ such that the sum gets
\begin{align*}
&\phantom{=} \binom{n}{k}  sk \sum_{j=1}^{n-s+1} \binom{n}{j}^{-1} \frac{1}{j(j-1+\rho)} \sum_{m=1}^{k-s+1}  \binom{k-s+1}{m}   \binom{n-k}{j-m}   \\
% & = \binom{n}{k}  sk \sum_{j=1}^{n-s+1} \binom{n}{j}^{-1} \frac{1}{j(j-1+\rho)} \left( \binom{n-s+1}{j} - \binom{n-k}{j} \right) .   \\  DAS IST FALSCH??
& = \binom{n}{k}  sk \sum_{j=1}^{n-s+1} \binom{n}{j}^{-1} \frac{1}{j(j-1+\rho)} \binom{n-k}{j-1} .   \\
\end{align*}

So we end up with
\begin{align*}
 \mathbb E[G_{k,s}] = & \frac{\theta_1}{k} \frac{k!}{(n-k+\rho)\cdots(n-1+\rho)} \frac{\theta_2}{s} \binom{n}{k}  \frac{k}{n} \binom{n-1}{s}^{-1} \sum_{j=0}^{n-s-1} \frac{j+1}{j+1+\rho} \binom{n-j-2}{s-1}\\
%   \frac{sk}{(k-s)(k-s+1)}  \\
%     & \quad \sum_{j=1}^{n-s+1} \binom{n}{j}^{-1} \frac{1}{j(j-1+\rho)} \sum_{m=1}^{k-s+1}   m (m-1)   \binom{k-s+1}{m}   \binom{n-k}{j-m}   \\
                    = &  \frac{\theta_1}{k} \frac{(n-k+1)\cdots n}{(n-k+\rho)\cdots(n-1+\rho)} \frac{\theta_2}{s} \frac{k}{n} \binom{n-1}{s}^{-1} \sum_{j=0}^{n-s-1} \frac{j+1}{j+1+\rho} \binom{n-j-2}{s-1},
%  \frac{sk}{(k-s)(k-s+1)} \\
%     & \quad \sum_{m=1}^{k-s+1} \sum_{j=1}^{n-k+m}  \binom{n}{j}^{-1} \frac{1}{j(j-1+\rho)}   m (m-1)   \binom{k-s+1}{m}   \binom{n-k}{j-m}   \\
\end{align*}
if $s < k$ and 
\begin{align*}
 \mathbb E[G_{k,s}] 
 = &  \frac{\theta_1}{k} \frac{(n-k+1)\cdots n}{(n-k+\rho)\cdots(n-1+\rho)} \frac{\theta_2}{s} sk  \sum_{j=1}^{n-s+1}  \binom{n}{j}^{-1} \frac{1}{j(j-1+\rho)}   \binom{n-k}{j-1}   \\
\end{align*}
if $s = k$ and we are done.

\end{proof}

\subsection{Proof of Corollary \ref{condsitefreqspec}}

We will proof Corollary \ref{condsitefreqspec} using Theorem \ref{jointfs}.
\begin{proof}[Proof of Corollary \ref{condsitefreqspec} based on Theorem \ref{jointfs}]
Note that we can write $\mathbb E [G_k] = \int_I \mathbb P(F(u) = k) du$, where the integrand does not depend on $u$ as  
\begin{equation*}
\mathbb P( F(u) = k) du = \sum_{i=1}^n \sum_{l=1}^i \mathbb P[\text{$l$th line during $T_i$ is of size $k$}] \mathbb P[\text{mark in $du$ on $l$th line during $T_i$}],
\end{equation*}
and marks in $du$ occur at rate $\tfrac{\theta_1 du}{2}$ independently of the other marks.
If we define 
\begin{equation*}
S^u_{s,dv} := |\{v \in dv : v \in \mathcal S^u_i \text{ for exactly $s$ different $i$ with $\mathcal M_i(u,0) = 1$}   \}| ,
\end{equation*}
we can be sure that $S^u_{s,dv} \leq 1$, as $dv$ is small and write
\begin{equation*}
\mathbb E[G_{k,s}] = \int_I \int_J  \mathbb P ( S^u_{s,dv} = 1  \cap F(u) = k) du.
\end{equation*}
So we end up with
\begin{align*}
\mathbb E[S^u_s\ |F(u)=k] &= \int_J \mathbb P ( S^u_{s,dv} = 1  \ | F(u) = k)
= \int_J \frac{ \mathbb P ( S^u_{s,dv} = 1  \cap F(u) = k)}{\mathbb P(F(u) = k) } \\
&= \frac{ \int_I  \int_J \mathbb P ( S^u_{s,dv} = 1  \cap F(u) = k) du}{ \int_I \mathbb P(F(u) = k) du}
= \mathbb E[G_{k,s}] \mathbb E[G_k]^{-1}.
\end{align*}
If $s < k$ we have
\begin{align*}
\mathbb E[&S^u_s\ |F(u)=k] = \mathbb E[G_{k,s}] \mathbb E[G_k]^{-1} = \frac{\theta_2}{s} \frac{k}{n} \binom{n-1}{s}^{-1} \sum_{j=0}^{n-s-1} \frac{j+1}{j+1+\rho} \binom{n-j-2}{s-1}.
% &=  \frac{\theta_2}{s} \sum_{j=1}^n \frac{1}{j(j-1+\rho)} \sum_{m=1}^k m \frac{s!(m-1)\cdots(k-s-1)}{(m)\cdots(k-1)}  \binom{k-m}{s-1}   \binom{n-j}{k-m}  \binom{j}{m}\\   
% DIESE ZEILE WAR DAS ALTE RESULTAT &=  \frac{\theta_2}{s} \frac{sk}{(k-s)(k-s+1)} \sum_{m=1}^{k-s+1} \sum_{j=1}^{n-k+m}  \binom{n}{j}^{-1} \frac{m (m-1)}{j(j-1+\rho)}    \binom{k-s+1}{m}   \binom{n-k}{j-m} \\
% &= \frac{k!(n-k)!}{n!} 
%     \sum_{m=1}^k m \frac{\theta_2}{s} \frac{s!(m-1)\cdots(k-s-1)}{(m)\cdots(k-1)} \binom{k-m}{s-1}  \sum_{i=1}^n  %\binom{n-i}{k-1}  \binom{n-i}{k-1}^{-1}
%                                             \sum_{j=i}^n \binom{j-i}{m-1} \binom{n-j}{k-m} \frac{1}{j(j-1+\rho)}\\
% &= \frac{k!(n-k)!}{n!} \frac{\theta_2}{s} \sum_{j=1}^n \frac{1}{j(j-1+\rho)} \sum_{m=1}^k m \frac{s!(m-1)\cdots(k-s-1)}{(m)\cdots(k-1)}  \binom{k-m}{s-1}   \binom{n-j}{k-m}  \binom{j}{m}\\ 
\end{align*}

In particular, if $k=n$ we get
\begin{equation*}
\mathbb E[S^u_s\ |f(u) = n] =  \frac{\theta_2}{s} \binom{n-1}{s}^{-1} \sum_{j=0}^{n-s-1} \frac{j+1}{j+1+\rho} \binom{n-j-2}{s-1}.
\end{equation*}

And for $s = k < n$ we have
\begin{align*}
\mathbb E[S^u_s\ |f(u)=k] &= \mathbb E[G_{k,s}] \mathbb E[G_k]^{-1}\\
&=  \frac{\theta_2}{s} sk \sum_{j=1}^{n-k+1}  \binom{n}{j}^{-1} \frac{1}{j(j-1+\rho)}  \binom{n-k}{j-1} \\
\end{align*}
and $ \mathbb E[S^u_n\ |f(u)=n] = \frac{\theta_2}{\rho}$ for $s = k = n$.
\end{proof}

\begin{remark}
For $\rho = 0$ we may proof Corollary \ref{condsitefreqspec} based on results from \cite{Wiuf1999} and \cite{Griffiths2003}.
Therefore apply Lemma 4 to Lemma 2 in \cite{Wiuf1999}, such that we get for an $k$-subtree beneath a gene gain event:
\begin{equation}
\mathbb P[\text{subtree has $m$ lineages, when the $n$-tree has $j$ lineages}] = \binom{n-j}{k-m}\binom{j}{m}\binom{n}{k}^{-1}.
\end{equation}
Let $T_m^\#$ be the time, where the subtree has $m$ lineages.
Then
\begin{equation}
\mathbb E[T_m^\#] = \sum_{j=m}^{n}  \binom{n-j}{k-m}\binom{j}{m}\binom{n}{k}^{-1} \mathbb E[T_j] = \frac{k}{n} \frac{2}{m(m-1)}.
\end{equation}
We can now compute, just as in equation (4.5) in \cite{Griffiths2003}, 
\begin{equation}
\mathbb E[S_s^u| F(u) = k] = \frac{\theta_2}{2} \sum_{m=2}^{k-s+1} m p_{k,m}(s) \mathbb E[T_m^\#] = \frac{k}{n}\frac{\theta_2}{s}.
\end{equation}
Here $p_{k,m}(s)$ is the probability for a particular lineage 
among $m$ lineages to be of size $s$ in a tree of size $k$, which we derived in
% Note that $p_{k,m}(s)$ 
equation \eqref{newhoppe}.
\end{remark}

\subsection{Proof of Theorem \ref{dispestimates}}

The upper bound for the estimators in Theorem \ref{dispestimates} is a simple consequences of the following Corollary.

\begin{corollary}
 Consider the site frequency spectrum $(S^u_s)_{s=1,\dots,k-1}$ for a dispensable gene $u$ as defined in Definition \ref{def:sfsdispensablegenes} and 
the site frequency spectrum $(C^u_s)_{s=1,\dots,k-1}$ for an essential core gene. Then
\begin{equation}
\mathbb E[S^u_s\ | F(u) = k] \leq \frac{k}{n} \mathbb E [C^u_s] = \frac{k}{n} \frac{\theta_2}{s}.
\end{equation}
\end{corollary}
\begin{proof}
 We have to show
\begin{equation}
\frac{\theta_2}{s} \frac{k}{n} \binom{n-1}{s}^{-1} \sum_{j=0}^{n-s-1} \frac{j+1}{j+1+\rho} \binom{n-j-2}{s-1}\\   < \frac{k}{n} \frac{\theta_2 }{s}  \label{sfsabschaetz}
\end{equation}
so it suffices to confirm
\begin{align*}
 & \frac{\theta_2}{s} \frac{k}{n} \binom{n-1}{s}^{-1} \sum_{j=0}^{n-s-1} \frac{j+1}{j+1+\rho} \binom{n-j-2}{s-1}\\   
 &\leq \frac{\theta_2}{s} \frac{k}{n} \binom{n-1}{s}^{-1} \sum_{j=0}^{n-s-1} \binom{n-j-2}{s-1}
 = \frac{\theta_2}{s} \frac{k}{n}. \qedhere
 \end{align*}
\end{proof}

\begin{proof}[Proof of Theorem \ref{dispestimates}]

From \eqref{sfsabschaetz} we see immediately
\begin{align*}
 \mathbb E[S_{\text{seg}}] &= \sum_{s=1}^{k-1} \mathbb E[ S_s^u | F(u) = k  ] \leq \frac{k}{n} \sum_{s=1}^{k-1} \mathbb E[ C_s^u ] = \frac{k}{n} \theta_2 \sum_{s=1}^{k-1} \frac{1}{s}\\
 \mathbb E[\widehat \theta_{W,k}] &\leq \frac{k}{n} \theta_2\\
 \mathbb E[\widehat \pi_k] &=  \sum_{s=1}^{k-1} \mathbb E[S_s^u| F(u) = k] \frac{s(k-s)}{\binom{k}{2}} \leq \sum_{s=1}^{k-1} \mathbb E[C_s^u] \frac{s(k-s)}{\binom{k}{2}} = \frac{k}{n} \theta_2\\ 
\end{align*}

For the estimators $\widehat \theta_{W,k}$ and  $\widehat \pi_k$ we set in the result from Corollary \ref{condsitefreqspec}
\begin{align*}
 \mathbb E[\widehat \theta_{W,k}] &= \frac{\sum_{s=1}^{k-1} \mathbb E[S_s^u| F(u) = k]  }{\sum_{s=1}^{k-1} \frac{1}{s}}\\
 &= \theta_2 \frac{\frac{k}{n} \sum_{s=1}^{k-1} \frac{1}{s} \binom{n-1}{s}^{-1} \sum_{j=0}^{n-s-1} \frac{j+1}{j+1+\rho} \binom{n-j-2}{s-1}}{\sum_{s=1}^{k-1} \frac{1}{s}}
\intertext{and}
 \mathbb E[\widehat \pi_k] &=  \sum_{s=1}^{k-1} \mathbb E[S_s^u| F(u) = k] \frac{s(k-s)}{\binom{k}{2}} \\
 &= \frac{\theta_2}{\binom{k}{2}} \frac{k}{n} \sum_{s=1}^{k-1} (k-s) \binom{n-1}{s}^{-1} \sum_{j=0}^{n-s-1} \frac{j+1}{j+1+\rho} \binom{n-j-2}{s-1} \\
\end{align*}
and we are done.

% ALTE RECHNUNG
% % Such that given a dispensable gene $u$, which appears in $k$ of $n$ individuals within the sample, the expected value again shrinks at least to
% \begin{align*}
%  \mathbb E_{disp}[\widehat \pi] &=  \sum_{s=1}^{k-1} \mathbb E[S_s^u| F(u) = k] \frac{s(k-s)}{\binom{k}{2}} \\
% 	&= \theta_2 \frac{2}{k(k-1)} \sum_{s=1}^{k-1} \frac{sk}{(k-s+1)} \sum_{m=1}^{k-s+1} \sum_{j=1}^{n-k+m}  \binom{n}{j}^{-1} \frac{m (m-1)}{j(j-1+\rho)}    \binom{k-s+1}{m}   \binom{n-k}{j-m} \\
%         &= \theta_2  \frac{2}{k(k-1)} k \sum_{m=2}^{k} \sum_{j=1}^{n-k+m}  \binom{n}{j}^{-1} \frac{m (m-1)}{j(j-1+\rho)}  \binom{n-k}{j-m} \sum_{s=m}^{k}   \binom{s}{m} \frac{k-s+1}{s} \\
%         &= \theta_2  \frac{2}{k-1} \sum_{m=2}^{k} \sum_{j=1}^{n-k+m}  \binom{n}{j}^{-1} \frac{m (m-1)}{j(j-1+\rho)}  \binom{n-k}{j-m} \binom{k}{m} \frac{1+k}{m(m+1)} \\
%         &= \theta_2  \frac{2}{k-1} \sum_{m=2}^{k} \sum_{j=0}^{n-k}  \binom{n}{j+m}^{-1} \frac{(m-1)}{(m+1)}\frac{(1+k)}{(j+m)(j+m-1+\rho)}  \binom{n-k}{j} \binom{k}{m}\\
%         &\leq \theta_2 \frac{2}{k-1} \sum_{m=2}^{k} \sum_{j=0}^{n-k}  \binom{n}{j+m}^{-1} \frac{(m-1)}{(m+1)}\frac{(1+k)}{(j+m)(j+m-1)}  \binom{n-k}{j} \binom{k}{m}\\
% 	&= \theta_2  \frac{k+1}{k-1} 2 \sum_{m=2}^{k} \frac{1}{m(m+1)} \binom{n}{n-k}^{-1} \sum_{j=0}^{n-k}   \binom{j+m-2}{j} \binom{n-j-m}{n-j-k}\\
% 	&=  \frac{k}{n} \theta_2 \frac{k+1}{k-1} 2 \sum_{m=2}^{k} \frac{1}{m(m+1)} = \frac{k}{n} \theta_2. \\
% \end{align*}

\end{proof}

\begin{section}{Acknowledgments}
 I thank Peter Pfaffelhuber, Joachim Hermisson and Iulia Dahmer for fruitful discussion and helpful comments.
 The German Research Foundation is acknowledged for funding via the project PP672/2-1.
\end{section}

\pagebreak

\bibliographystyle{chicago}
\bibliography{newbib.bib}

\end{document}